\documentclass[letterpaper,11pt]{scrartcl}
\usepackage[utf8]{inputenc}
\usepackage[svgnames]{xcolor}
\usepackage{authblk}
\usepackage{amsmath,amssymb,csquotes,xcolor,array,multicol,amsthm,thmtools}
\usepackage{enumitem,xparse}
\typearea{15}

\usepackage[margin=1in]{geometry}
\usepackage{graphicx}
\usepackage{times}

\usepackage{enumitem}
\setlist{nolistsep}

\usepackage{microtype}
\usepackage{verbatim}
\usepackage{hyperref}

\usepackage{xspace}

\usepackage[plain]{algorithm}
\usepackage{algpseudocode}
\algtext*{EndWhile}
\algtext*{EndFor}
\algtext*{EndProcedure}
\algtext*{EndIf}
\algnewcommand\algorithmicinput{\textbf{Input:}}
\algnewcommand\Input{\item[\algorithmicinput]}
\algnewcommand\algorithmicoutput{\textbf{Output:}}
\algnewcommand\Output{\item[\algorithmicoutput]}
\newcommand{\algparbox}[1]{\parbox[t]{\dimexpr\linewidth-\algorithmicindent}{#1\strut}}
\newcommand{\StateN}[1]{\State \algparbox{#1}}

\graphicspath{{figs/}}

\newtheorem{definition}{Definition}
\newtheorem{lemma}[definition]{Lemma}
\newtheorem{observation}[definition]{Observation}
\newtheorem{proposition}[definition]{Proposition}
\newtheorem{corollary}[definition]{Corollary}
\newtheorem{theorem}[definition]{Theorem}

\newtheorem{claim}[definition]{Claim}

\renewcommand{\lg}{\log}
\renewcommand{\varepsilon}{\epsilon}

\newcommand{\poly}{\operatorname{poly}}

\renewcommand{\P}{\mathcal{P}}

\newcommand{\NP}{\textsf{NP}}

\DeclareRobustCommand{\ALG}{%
	\ifmmode
		\operatorname{ON}
	\else
		\text{ON}\xspace
	\fi
}
\DeclareRobustCommand{\OFF}{%
	\ifmmode
		\operatorname{OFF}
	\else
		\text{OFF}\xspace
	\fi
}
\DeclareRobustCommand{\APPROXALGO}{%
	\ifmmode
		\operatorname{APPROX}
	\else
		\text{APPROX}\xspace
	\fi
}

\def\bold #1{{\bfseries\mathversion{bold}#1}}

\def\OPT{\operatorname{OPT}}
\def\ONL{\operatorname{ONL}}
\def\cost{\operatorname{cost}}
\def\size{\operatorname{vol}}

\title{Tight Bounds for Online Graph Partitioning\thanks{This is the full version of a paper which will appear at SODA'21.}}

\author[1]{Monika Henzinger\thanks{monika.henzinger@univie.ac.at.}}
\author[2]{Stefan Neumann\thanks{neum@kth.se. This work was done while Stefan Neumann was at University of Vienna.}}
\author[3]{Harald R\"acke\thanks{raecke@in.tum.de}}
\author[1]{Stefan Schmid\thanks{stefan\_schmid@univie.ac.at.}}
\affil[1]{University of Vienna, Austria}
\affil[2]{KTH Royal Institute of Technology, Sweden}
\affil[3]{TU Munich, Germany}
\date{}

\begin{document}

\pagenumbering{roman}
\maketitle

\begin{abstract}
We consider the following online optimization problem. We are given a graph $G$
and each vertex of the graph is assigned to one of $\ell$ servers, where servers
have capacity $k$ and we assume that the graph has $\ell \cdot k$ vertices.
Initially, $G$ does not contain any edges and then the edges of $G$ are revealed
one-by-one. The goal is to design an online algorithm $\ONL$, which always
places the connected components induced by the revealed edges on the same server
and never exceeds the server capacities by more than $\varepsilon k$ for
constant $\varepsilon>0$.  Whenever $\ONL$ learns about a new edge, the
algorithm is allowed to move vertices from one server to another. Its objective
is to minimize the number of vertex moves. More specifically, $\ONL$ should
minimize the competitive ratio: the total cost $\ONL$ incurs compared to an
optimal offline algorithm $\OPT$.

The problem was recently introduced by Henzinger et al.\ (SIGMETRICS'2019) and
is related to classic online problems such as online paging and scheduling. It
finds applications in the context of resource allocation in the cloud and for
optimizing distributed data structures such as union--find data structures.

Our main contribution is a polynomial-time randomized algorithm, that is
asymptotically optimal: we derive an upper bound of $O(\log \ell + \log k)$ on
its competitive ratio and show that no randomized online algorithm can achieve a
competitive ratio of less than $\Omega(\log \ell + \log k)$.  We also settle the
open problem of the achievable competitive ratio by deterministic online
algorithms, by deriving a competitive ratio of $\Theta(\ell \lg k)$; to this
end, we present an improved lower bound as well as a deterministic
polynomial-time online algorithm.

Our algorithms rely on a novel technique which combines efficient integer
programming with a combinatorial approach for maintaining ILP solutions. More
precisely, we use an ILP to assign the connected components induced by the
revealed edges to the servers; this is similar to existing approximation schemes
for scheduling algorithms.  However, we cannot obtain our competitive ratios if
we run the ILP after each edge insertion. Instead, we identify certain types of
edge insertions, after which we can manually obtain an optimal ILP solution at
zero cost without resolving the ILP.  We believe this technique is of
independent interest and will find further applications in the future.

\end{abstract}

\pagebreak
\tableofcontents
\pagebreak

\pagenumbering{arabic}

\def\hmax{h_{\max}}

\def\TypeI  {Type~I\xspace}
\def\TypeII {Type~II\xspace}
\def\TypeIII{Type~III\xspace}
\def\TypeIV {Type~IV\xspace}
\def\TypeV  {Type~V\xspace}

\section{Introduction}
\label{onl:sec:introduction}

Distributed cloud applications 
generate a significant amount of network traffic~\cite{talk-about}. 
To improve their efficiency and performance,
the underlying infrastructure needs to become \emph{demand-aware}: 
frequently communicating 
endpoints need to be allocated closer to each other, e.g., by
collocating them on the same server or in the same rack.
Such optimizations are enabled by the increasing 
resource allocation flexibilities available
in modern virtualized infrastructures, and 
are further motivated by
rich spatial and temporal structure featured 
by communication patterns of data-intensive applications~\cite{Benson-imc}.

This paper studies the algorithmic problem underlying 
such demand-aware resource allocation in scenarios
 where the communication pattern is not known ahead of time.
Instead, the algorithm needs to  learn
a communication pattern in an online manner, dynamically collocating communication partners 
while minimizing reconfiguration costs.
It has recently been shown that this problem can be modeled 
by the following \emph{online graph partitioning problem}~\cite{henzinger19efficient}:
Let $k$ and $\ell$ be known parameters. We are
given a graph $G$ which initially does not contain any edges. Each
vertex of the graph is assigned to one of $\ell$ servers and each 
server has \emph{capacity} $k$, i.e., stores $k$ vertices. 
Next, the edges of $G$ are revealed one-by-one in an online fashion
and the algorithm has to guarantee that \emph{every connected component 
is placed on the same server (cc-condition).} This is possible, as
it is guaranteed that there always exists an assignment of the connected components to servers such that no connected component is split across multiple
servers.  
Thus after each edge insertion, the online algorithm has to decide which vertices to move between servers to guarantee the cc-condition. Each vertex move incurs a cost of $1/k$. 
The optimal offline algorithm knows all the connected components, computes, dependent on the initial placement of the vertices, the minimum-cost assignment of these connected components to servers and moves the vertices to their final servers after the first edge insertion. It requires no further vertex moves. To measure the performance of an online algorithm $\ONL$ we use 
the competitive ratio: the total cost of $\ONL$ divided by the total cost of the
optimal offline algorithm $\OPT$.
In this setting, no deterministic online algorithm can
have a competitive ratio better than $\Omega(\ell \cdot k)$ even with unbounded
computational resources~\cite{henzinger19efficient,podc20ba}.
In fact, even assigning the connected components to the servers such that the
capacity constraints are obeyed is already \NP-hard and this holds even in the
static setting when the connected components do not change~\cite{andreev2006balanced}.

Thus, we relax the server capacity requirement for the online algorithm: Specifically, the online algorithm is allowed to place
up to $(1 + \varepsilon) k$ vertices on a server at any point in time. We call this problem the \emph{online graph partitioning problem}.

Henzinger et al.~\cite{henzinger19efficient} studied this problem and showed that the previously
described demand-aware resource allocation problem reduces to the online graph
partitioning problem: vertices of the graph $G$ correspond to communication
partners and edges correspond to communication requests; thus, by collocating
the communication partners based on the connected components of the vertices, we
minimize the network traffic (since all future communications among the revealed
edges will happen locally). They also showed how to implement a distributed
union--find data structure with this approach.
Algorithmically, \cite{henzinger19efficient}~presented a deterministic
exponential-time algorithm with competitive ratio $O(\ell \log \ell \log k)$ and
complemented their result with a lower bound of
$\Omega(\log k)$ on the competitive ratio of any
deterministic online algorithm.
While their derived bounds are tight for $\ell=O(1)$ servers, 
there remains a gap of factor $O(\ell \lg \ell)$ between upper and lower bound
for the scenario of $\ell=\omega(1)$.
Furthermore, their lower bound only applies to deterministic algorithms and thus
it is a natural question to ask whether randomized algorithms can obtain
better competitive ratios.

\subsection{Our Contributions}

Our main contribution is a polynomial-time randomized algorithm
for online graph partitioning which achieves a polylogarithmic
competitive ratio. 
In particular, we derive an $O(\log \ell + \log k)$ upper bound on the competitive ratio of 
our algorithm, where
$\ell$ is the number of servers and $k$ is the server capacity. 
We also show that no randomized online algorithm can achieve
a competitive ratio of less than $\Omega(\log \ell + \log k)$.
The achieved competitive ratio is hence asymptotically optimal.

We further settle the open problem of the competitive ratio achievable
by deterministic online algorithms. To this end, we derive an improved lower
bound of $\Omega(\ell \lg k)$, and present a polynomial-time deterministic
online algorithm which achieves a competitive ratio of $O(\ell \lg k)$. Thus,
also our deterministic algorithm is optimal up to constant factors in the competitive
ratio.

These results improve upon the results of~\cite{henzinger19efficient}
in three respects: First, our deterministic online algorithm has competitive ratio
$O(\ell \log k)$ and polynomial run-time, while the algorithm
in~\cite{henzinger19efficient} has competitive ratio $O(\ell \lg \ell \lg k)$
and requires exponential time.  Second, we present a significantly higher and
matching lower bound of $\Omega(\ell \log k)$ on the competitive ratio of any
deterministic online algorithm.  Third, we initiate the study of
randomized algorithms for the online graph partitioning problem and show that it
is possible to achieve a competitive ratio of $O(\lg\ell+\lg k)$ and we
complement this result with a matching lower bound.  Note that the competitive
ratio obtained by our randomized algorithm provides an exponential improvement
over what any deterministic algorithm can achieve in terms of the dependency on
the parameter $\ell$.

We further show that for $\varepsilon>1$ (i.e., the servers can store at least
		$(2+\varepsilon')k$ vertices for some $\varepsilon'>0$), our
deterministic algorithm is $O(\lg k)$-competitive.
 
\textbf{Technical Novelty.}
We will now provide a brief overview of 
our approach and its technical novelty. 
Since our deterministic and our randomized 
algorithms are based on the same algorithmic framework, we will say \emph{our algorithm} in the following. 

Our algorithm keeps track of the set of connected components induced by the
revealed edges.  
We will denote the connected
components as \emph{pieces} and when two connected components become connected
due to an edge insertion, we say that the corresponding pieces are
\emph{merged}.

The algorithm maintains an assignment of the pieces onto the servers
which we call a \emph{schedule}. We will make sure that the schedule is
\emph{valid}, i.e., that every piece is assigned to some server and that the
capacities of the servers are never exceeded by more than the allowed
additive $\varepsilon k$.  To compute valid schedules, we
solve an integer linear program (ILP) using a generic ILP solver and show how the solution of the ILP can
be transformed into a valid schedule. We ensure that the ILP is of constant size
and can, hence, be solved in polynomial time.  Next, we show that when two
pieces are merged due to an edge insertion, the schedule does not change
much, i.e., we do not have to move ``too many'' pieces between the servers.  We do
this using a \emph{sensitivity analysis} of the ILP, which guarantees that when
two pieces are merged, the solution of the ILP does not change by much. 
Furthermore, we prove that this 
change in the ILP solution 
corresponds to only slightly adjusting the schedules,
and thus only moving a few pieces.

However, the sensitivity analysis alone is not enough to obtain the desired competitive
ratio. Indeed, we identify certain types of merge-operations for which the
optimal offline algorithm $\OPT$ might have 
very small or even zero cost. In
this case, adjusting the schedules based on the ILP sensitivity would be too
costly: the generic ILP solver from above could potentially move to an optimal ILP solution which is very
different from the current solution and, thus, incur much more cost than $\OPT$. 
Hence, to keep the cost paid by our algorithm low, we make sure that for 
these special types of merge-operations, our algorithm sticks extremely close to
the previous ILP solution, incurs zero cost for moving pieces and
still obtains an optimal ILP solution.
The optimality of the algorithm's solution after such merge-operations is crucial as
otherwise, we could not apply the sensitivity analysis after the subsequent merge-operations. 
To the best of our knowledge, our algorithm is the first to \emph{interleave ILP sensitivity analysis 
with manual maintenance of optimal ILP
solutions.}

More specifically, we assume that each server has a unique color and consider each vertex as being colored by the color of its initial server.
In our analysis we identify two different types of pieces:
\emph{monochromatic} and \emph{non-monochromatic} ones. In the monochromatic pieces,
``almost all'' of the vertices  have the same color, i.e.,~were initially 
assigned to the same server, while the non-monochromatic pieces contain ``many''
vertices which started on different servers. We show that we have to treat the
monochromatic pieces very carefully because these are the pieces
for which $\OPT$ might have very small or
even no cost. Hence, it is desirable
to always schedule monochromatic pieces on the server of the majority color.
Unfortunately, we show that this is not always possible.
Indeed, the hard instances in our lower bounds show that an adversary can force
any deterministic algorithm to create schedules with \emph{extraordinary}
servers. Informally, a server $s$ is extraordinary if there exists a monochromatic
piece $p$ for which almost all of its vertices have the color of $s$
but $p$ is not scheduled on $s$ (see Section~\ref{onl:sec:ilp} for the formal
definition). All other servers are 
\emph{ordinary} servers. Note that we have to deal with extraordinary servers carefully: 
we might have to pay much more than $\OPT$ for their
monochromatic pieces that are scheduled on other servers.

Thus, to obtain a competitive algorithm, we need to minimize the
number of extraordinary servers. We achieve this with the following idea:
We equip our ILP with an objective function that minimizes the number of
extraordinary servers and we show that 
the number of extraordinary servers created by the
ILP gives a lower bound on the cost paid by $\OPT$. We use this fact to argue that
we can charge the algorithm's cost when creating extraordinary servers to $\OPT$
to obtain competitive results.

The previously described ideas provide a deterministic algorithm with
competitive ratio $O(\ell \lg k)$. We also provide a matching lower bound of
$\Omega(\ell \lg k)$. The lower bound provides a hard instance which essentially
shows that an adversary can force any deterministic algorithm to make each of
the $\ell$ servers extraordinary at some point in time.  
More generally, the adversary can cause such a high competitive ratio whenever it knows \emph{which
servers are extraordinary.}
Hence, to obtain an
algorithm with competitive ratio $O(\lg\ell + \lg k)$ we use randomization to 
keep the adversary from knowing the extraordinary servers.

Our strategy for picking the extraordinary servers randomly is as follows.
First, we show that our algorithm experiences low cost (compared to $\OPT$) when
two pieces assigned to an ordinary
server are merged, while the cost for merging
pieces that are assigned to extraordinary servers is large (compared to $\OPT$). Next, we
reduce the problem of picking extraordinary servers to a paging problem, where
the pages correspond to servers such that pages in the cache
correspond to ordinary servers and pages outside the cache correspond to
extraordinary servers. Now when two pieces are merged, we issue the
corresponding page requests: A merge of pieces assigned to an ordinary
server corresponds to a page request of a page which is inside the cache, while
merging two pieces with at least one assigned to an extraordinary server corresponds to
a page request of a page which is not stored in the cache.  This leads the
paging algorithm to insert and evict pages into and from the cache and our
algorithm always changes the types of the servers accordingly. For example, when
a page corresponding to an ordinary server is evicted from the cache, we make
the corresponding server extraordinary. We conclude by showing that since the
randomized paging algorithm of Blum et al.~\cite{finely-competitive-paging} allows for a
polylogarithmic competitive ratio, we also obtain a polylogarithmic competitive
ratio for our problem.

\textbf{Future Directions.}
While the algorithms presented in this paper are asymptotically optimal (i.e.,
our upper and lower bounds on the competitive ratios match), our work opens
interesting avenues for future research. A first intriguing question for future
research regards the study of scenarios without the cc-condition.      
In this
setting, Avin et al.~\cite{obr-original,avin19dynamic} provided a deterministic
$O(k \lg k)$-competitive algorithm, \textsc{Crep}, for servers with capacity $(2+\varepsilon)k$
and complemented this result with a $\Omega(k)$ lower bound for deterministic
algorithms; while  \textsc{Crep} had a super-polynomial runtime, Forner et al.
\cite{apocs21repartition} recently demonstrated a polynomial-time
implementation, \textsc{pCrep}, which monitors the connectivity of communication
requests over time, rather than the density as in \textsc{Crep}, and this
enables a faster runtime.  However, nothing is known for randomized algorithms
and it would be exciting to determine whether polylogarithmic competitive ratios
are achievable.  Second, the competitive ratios of our algorithms depend
exponentially on $1/\varepsilon$ (see the discussions after
Theorems~\ref{onl:thm:det} and~\ref{onl:thm:randomized}). It would be
interesting to obtain tight competitive ratios with only a polynomial dependency
on $\varepsilon$, as this might yield more practical algorithms. To realize this
goal, it seems that one first has to come up with a PTAS for scheduling on
parallel identical machines in the \emph{offline} setting, where the dynamic
program (DP) or the ILP formulation has only $\poly(1/\varepsilon)$ DP cells or
variables, respectively.  In particular, this rules out using the technique by
Hochbaum and Shmoys~\cite{hochbaum87using}, that our algorithm relies on.

\subsection{Related Work}
\label{onl:sec:related}

The online graph partitioning problem considered in this paper
is generally related to classic online problems
such as competitive paging and caching~\cite{paging-mark,companion-caching,competitive-analysis,young-paging-soda,generalized-caching-optimal,caching-rejection-penalties},
$k$-server~\cite{k-server2}, or 
metrical task systems~\cite{metrical-task-systems}.
However, unlike these existing problems
where requests are typically related to specific items
(e.g., in paging) or locations in a graph or metric space
(e.g., the $k$-server problem and metrical task systems),
in our model,
requests are related to \emph{pairs of vertices}.
The problem can hence also be seen as a \emph{symmetric}
version of online paging, where each of the two items (i.e., vertices in our model)
involved in a request can be moved to either of the servers currently hosting one
of the items (or even to a third server).
The offline problem variant
is essentially a
$k$-way partitioning or 
graph partitioning
problem~\cite{Vaquero:2013:APL:2523616.2525943,abbe2018community}. The balanced
graph partitioning problem is 
related to minimum bisection~\cite{feige2002}, and
known to be hard to 
approximate~\cite{andreev2006balanced,Krauthgamer2006}. 
Balanced clustering problems have also been studied in 
streaming settings~\cite{stanton2014streaming,alistarh15streaming}.

Our model is also related to 
dynamic bin packing problems which allow for limited \emph{repacking}~\cite{FeldkordFGGKRW18}:
this model can be seen as a variant of our problem where pieces (resp.~items)
can both be dynamically inserted and deleted, and it is also possible to open new
servers (i.e., bins); the goal is to use only an (almost) minimal number of servers, and
to minimize the number of piece (resp.~item) moves.
However, the techniques of~\cite{FeldkordFGGKRW18} do not extend to our problem.

Another related problem arises in the context of generalized online scheduling,
where the current server assignment can be changed whenever a new job 
arrives,
subject to the constraint that the total size of moved jobs is bounded by some
constant times the size of the arriving job. While the reconfiguration cost in
this model is fairly different
from ours, the sensitivity analysis of our ILP 
is inspired by the techniques used in
Hochbaum and Shmoys~\cite{hochbaum87using} and
Sanders et al.~\cite{sanders09online}.

Our work is specifically motivated by 
the online balanced (re-)partitioning problem
introduced by Avin et al.~\cite{obr-original,avin19dynamic}.
In their model, the connected components of the graph $G$ can contain more than
$k$ vertices and, hence, might have to be split across multiple servers.
They presented a lower bound of $\Omega(k)$ for deterministic algorithms.
They complemented this result by a deterministic algorithm with competitive
ratio of $O(k \log k)$. This problem was also studied when the graph $G$ follows
certain random graphs models~\cite{obr-ring,netys17learn}.
For a scenario without resource augmentation, i.e., $\varepsilon = 0$,
there exists an $O(\ell^2\cdot k^2)$-competitive algorithm \cite{avin19dynamic} 
and a lower bound of $\Omega(\ell\cdot k)$ \cite{podc20ba}.

\textbf{Organization.}
Section~\ref{onl:sec:preliminaries} introduces our notation and
Section~\ref{onl:sec:overview} gives an overview of the algorithmic framework.
We explain the deterministic algorithm in detail, including the ILP, in
Section~\ref{onl:sec:ilp} and analyze it in Section~\ref{onl:sec:analysis}.
Section~\ref{onl:sec:randomized-algorithm} presents the randomized algorithm.
Our lower bounds are presented in Section~\ref{onl:sec:lbs} and
Appendix~\ref{onl:sec:omitted} contains omitted proofs.

\def\dif{D}
\def\ps{p_{1}}
\def\pl{p_{2}}
\def\pm{p_{m}}
\def\is{i_{1}}
\def\il{i_{2}}
\def\im{i_{m}}

\section{Preliminaries}
\label{onl:sec:preliminaries}

Let us first re-introduce our model together with some definitions.
We are given a graph $G=(V, E)$ 
with $|V| =\ell \cdot k$. In the beginning, $E = \emptyset$ and then edges are inserted in an
online manner. Initially, every vertex $v$ is assigned to one of
$\ell$~servers such that each server is assigned exactly $k$ vertices.
We call this server the \emph{source server} of $v$. For a server
$s$ we call the vertices which are initially assigned to $s$ the \emph{source vertices}
of $s$. 
For normalization
purposes, we consider each vertex to have a \emph{volume} $\size(v)$ of $1/k$, so that the
total volume of vertices initially assigned to a server is exactly $1$.

After each edge insertion, the online
algorithm must re-assign vertices to fulfill the \emph{cc-condition}, i.e., so that all vertices
of the same connected component of $G$ are assigned to the same server. To this
end, it can move vertices between servers at a cost of $1/k$ per vertex move.
As described in the introduction, the optimum offline algorithm $\OPT$ is only allowed
to place vertices with total volume up to 1 onto each server, while the online
algorithm $\ONL$ is allowed to place total volume
of total volume
of $1+\epsilon$ on each server, where $\epsilon > 0$ is a small constant.
For notational convenience, we will place a total volume of $1+c\cdot \epsilon$
for some constant~$c$,
which does not affect our asymptotic results as the algorithm can be started with
 $\epsilon' = \epsilon/c$.

Formally, the objective is to devise an online algorithm $\ONL$ 
which 
minimizes the \emph{(strict) competitive ratio $\rho$} defined as
$\rho=\cost(\ONL)/\cost(\OPT)$, where $\cost(\cdot)$
denotes the total volume of pieces moved by the corresponding
algorithm. For deterministic online algorithms,
the  edge insertion order is adversarial; 
for randomized online algorithms, we assume
an oblivious adversary that fixes an adversarial
request sequence without knowing the random choices
of the online algorithm.

The following definitions and concepts
are used in the remainder of this paper. 

\textbf{Pieces.}
Our online algorithm proceeds by tracking the \emph{pieces}, the connected
components of $G$ induced by the revealed edges.  The \emph{volume} of a piece
$p$, denoted by $|p|$, is the total volume of all its vertices.
For convenience, every server has a unique color from the set $\{1, \dots, \ell\}$
and every vertex has the \emph{color} of the server it was \emph{initially} assigned
to. 
For a
piece $p$, we define the \emph{majority color} of $p$ as the color that appears
most frequently among vertices of $p$ and, in case of ties, that is the smallest in the
order of colors.
We also refer to the corresponding server as the \emph{majority server} of $p$.
Similarly, we define the \emph{majority color} for a vertex $v$ to be the majority color
of the piece of $v$.  Note that the latter changes
dynamically as the connected components of $G$ change due to edge insertions.

\textbf{Size Classes and Committed Volume.}
To minimize frequent and expensive moves, 
 our approach
groups the pieces into small and large pieces, and for the ILP also partitions them into a constant number of size classes. The basic idea is to
``round down'' the volume of a piece to a suitable multiple of $1/k$ and to call all pieces of zero rounded volume small. 
However, pieces can grow and, thus, change their size class, which in turn might
create cost for the  online algorithm. Thus, we need to use a more ``refined'' rounding, that gives us some control over when such a class change occurs. 

More formally,  let us assume that
$1/4 > \epsilon \ge (10/k)^{1/4}$. We choose $\delta$ such that
$\frac{1}{2}\epsilon^2\le\delta\le\epsilon^2$ and $\delta=j\frac{1}{k}$ for some 
$j\in\mathbb{N}$. In addition, we assume
$\lceil 1\rceil_\delta - 1 \leq \delta/2$, where $\lceil\cdot\rceil_\delta$ is
the operation of rounding up to the closest multiple of $\delta$.
Claim~\ref{onl:claim:picking-delta} in Section~\ref{onl:sec:claim:picking-delta} shows that we can always find
such a $\delta$ provided that $k\ge 10/\epsilon^4$. We will also use a
constant $\gamma = 2\delta < 1$.

We partition the volume of a piece into \emph{committed} 
and \emph{uncommitted}
volume. The committed volume will always be a multiple of $\delta$, while the
uncommitted volume will be rather small (see below). 
We refer to the sum of committed and uncommitted volume
as the \emph{real volume} of the piece. We extend this definition to vertices:
Each vertex is either committed or uncommitted. Now the committed volume of a
piece is the volume of its committed vertices.
For a piece $p$, we write $|p|_c$ to denote its committed volume and $|p|_u$ to
denote its uncommitted volume. Hence, $|p| = |p|_c + |p|_u$.

We introduce size classes for the pieces. We say that a piece
is in \emph{class $i\in\mathbb{N}$} if its \emph{committed} volume is $i\cdot\delta$ (recall that
committed volume is always a multiple of $\delta$). 
Since the volume of a piece is never larger than
$1$, we have that $i \leq 1/\delta$ and thus there are only
$O(1/\delta) = O(1/\varepsilon^2) = O(1)$ size classes in total.

\textbf{Large and Small Pieces.}
Intuitively, we want to refer to pieces with total volume at least $\epsilon$ as
\emph{large} and to the remaining pieces as \emph{small}. For technical reasons,
we change this as follows. A piece is \emph{large} if its committed
volume is non-zero and \emph{small} otherwise. Thus, the small
pieces are exactly the pieces in class $0$. As the algorithm decides when to commit
volume, it controls the transition from small to large. Note that committed volume
never becomes uncommitted and, hence, a piece transitions only once from small to large.

\textbf{Monochromatic Pieces.}
Pieces that overwhelmingly contain vertices of a single color have to be
handled very carefully by an online algorithm because $\OPT$ may
not have to move many vertices of such a piece and thus experience very little
cost. Therefore we introduce the following notion, which needs to be different for small and large pieces since we use different scheduling techniques for them:
A large piece is called \emph{monochromatic} for its majority server
$s$ if the volume of its vertices 
that did not originate at $s$ is at most $\delta$. A small
piece is called \emph{monochromatic} if an $\epsilon$-fraction of its volume
did not originate at the majority server of the piece. We refer to pieces that
are not monochromatic as \emph{non-monochromatic}.

\section{Algorithmic Framework}
\label{onl:sec:overview}
In this section we present our general algorithmic framework. Some further details follow in
Section~\ref{onl:sec:ilp}.

(1) The algorithm always maintains the current set of pieces $\P$, where each piece is 
annotated by its size class.
If an edge insertion {merges} two pieces $\ps$ and $\pl$,
into a new merged piece $\pm$, it holds that
$|\pm|_u = |\ps|_u+|\pl|_u$ and $|\pm|_c = |\ps|_c+|\pl|_c$.
We say that a merge is \emph{monochromatic} if all of $\ps$, $\pl$ and $\pm$ are
monochromatic for the same server $s$.
Throughout the rest of the paper, we assume w.l.o.g.\ that $|\ps|\leq|\pl$ and
we let $\pm$ denote the piece that
resulted from the \emph{last merge-operation}.

\textbf{Invariants for Piece Volumes.}
Whenever the algorithm has completed its vertex moves after an edge insertion,
the following invariants for piece volumes are maintained.

\begin{enumerate}[itemsep=0pt]
\item A piece $p$ is small (i.e. has $|p|_c = 0$) iff $|p|<\epsilon$.
\label{onl:inv:smallpiecesmall}
\label{onl:inv:largepiecelarge}
\item A large piece has committed volume $i\cdot\delta$ for some
      $i\in\mathbb{N}, i> 0$. If it is monochromatic, all committed
      volume must be from its majority color.
\label{onl:inv:integralcommit}
\item\label{onl:inv:uncommitted-volume} The uncommitted volume of a large
      piece is at most $2\delta$, while the uncommitted volume of a small piece is at most $\epsilon$. 
	  Note that $2\delta = O(1/\varepsilon^2) \ll \varepsilon$.
\end{enumerate}
Now suppose that before a merge-operation, all pieces fulfill the invariants.
Then after the merge, the new piece $\pm$ might fulfill only the relaxed
constraint $|\pm|< 2\epsilon$ if $\pm$ is small (no committed volume) and the relaxed 
constraint  $|\pm|_u\le\epsilon+2\delta$ if $\pm$ is large (with committed
volume).\footnote{The first case occurs when $\ps$ and $\pl$ are small
	and have volume just below $\varepsilon$ (since then
	$|\pm|_c = |\ps|_c+|\pl|_c=0$ and $|\pm|_u = |\ps|_u+|\pl|_u<2\varepsilon$);
	in this case, $\pm$ is initially small and will become large only later when
	commit-operations are performed. The second case occurs when $\pl$ is large with $|\pl|_u$
	just below $2\delta$ and $\ps$ is small with $|\ps|$ just below
	$\varepsilon$.}
Before the next merge-operation, we will
perform \emph{commit-operations} on the piece $\pm$ until $\pm$ fulfills the
above invariants.  More concretely, 
if  $|\pm| \ge \epsilon$ then a commit-operation is executed as long as $|\pm|_u > 2 \delta$.
It selects uncommitted vertices inside $\pm$ of volume $\delta$ and sets their
state to committed (which makes $\pm$ large). If $\pm$ is monochromatic, the
commit-operation only selects vertices of the majority color, of which there is
a sufficient number since for large monochromatic pieces, the volume of
vertices of non-majority color is at most $\delta$.

(2) The algorithm further maintains a \emph{schedule $S$}, which is an assignment
of the pieces in $\P$ to servers. The algorithm guarantees that this schedule
fulfills certain invariants---the most important being 
the fact that the
total volume of pieces assigned to a server does not exceed  the server's
capacity by much.

\textbf{Adjusting Schedules.}
To reestablish these invariants after a change to $\P$, we run the \emph{adjust
schedule subroutine}. We provide the details of this subroutine in
Section~\ref{onl:sec:ilp} and now give a very short summary.
When the set $\P$ changes, this is due to one of two reasons: a \emph{merge
operation} or a \emph{commit-operation}.  Both types of changes might force us
to change the old schedule $S$ to a new schedule $S'$. To do so, we first solve
an ILP (to be defined in Section~\ref{onl:sec:ilp}) that computes the rough structure of
the new schedule $S'$. The ILP solution defines, among other things, the number
of extraordinary servers in the new schedule $S'$. Then we determine a concrete
schedule $S'$ that conforms to the structure provided by the ILP solution.
Crucially, we have to determine an $S'$ that is not too different from $S$, in
order to keep the cost for switching from $S$ to $S'$ small.

We note that the subroutine for adjusting the schedules only moves pieces
between the servers and hence does not affect the invariants for piece volumes.

\textbf{Handling an Edge Insertion.}
We now give a high-level overview of the algorithm when an edge
$(u,v)$ is inserted.  If 
$u$ and $v$ are part of the same piece, 
we do nothing since 
$\P$ did not change.  Otherwise, assume that $u$ is in piece $\ps$ and $v$ is
in $\pl$ with $\ps\neq\pl$ and $|\ps|\leq |\pl|$. Then we proceed as follows.
\begin{description}
	\item[Step~I] \emph{Move small to large piece:}
	Move the smaller piece $\ps$ to the server of the larger piece $\pl$.

	\item[Step~II] \emph{Merge pieces:}
	Merge $\ps$ and $\pl$ into $\pm$.
	Run the adjust schedule subroutine.

	\setlength{\intextsep}{1\baselineskip}
	\item[Step~III] \emph{Commit volume:}
		If $|\pm|\geq\varepsilon$, then
		\noindent
		\intextsep0pt
		\begin{algorithmic}
		\While{$|\pm|_u > 2\delta$}
			\State Commit volume $\delta$ for $\pm$.  Run the adjust schedule subroutine.
		\EndWhile
		\end{algorithmic}
\end{description}

\section{Adjusting Schedules}
\label{onl:sec:ilp}
Now we describe the subroutine for adjusting schedules in full detail.  In the
following, we define an ILP that helps in finding a good assignment of the
pieces to servers. We ensure that when the set of pieces only changes slightly,
then also the ILP solution only changes slightly.  We also show how the ILP
solution can be mapped to concrete schedules.

Before we describe the ILP in detail, we introduce reservation and
source vectors, as well as configurations. In a nutshell, a server's reservation
vector encodes how many pieces of each size class can be
assigned to that server at most.
A server's source vector, on the other hand, describes
the structure of the monochromatic pieces for that server. A configuration
is a pair of a reservation and a source vector and solving
the ILP will inform us which configurations should be used for the servers in
our schedule.

A \emph{reservation vector $r_s$} for a server $s$ has the
following properties. For a size class $i>0$, the entry $r_{si}$ describes the
total volume reserved on $s$ for the (committed) volume of pieces in class $i$
(regardless of their majority color). The entry $r_{s0}$ describes the total
volume that is reserved for uncommitted vertices (again, regardless of color);
note that these uncommitted vertices could belong to small or large pieces.
An entry $r_{si}$, $i>0$, must
be a multiple of $i\delta$ while the entry $r_{s0}$ is a multiple of $\delta$.
Note that $r_s$ does \emph{not} describe which concrete pieces are scheduled on
$s$ and not even the exact number of pieces of a certain class, as it only
``reserves'' space.

A \emph{source vector} $m_s$ for server $s$ has the following properties.
For a size class $i>0$, the
entry $m_{si}$ specifies the total committed volume of pieces in
class $i$ \emph{that are monochromatic for $s$}. Again recall that a
monochromatic piece only has committed volume of its majority color.  The entry
$m_{s0}$ describes the total uncommitted volume of color $s$ rounded up to a
multiple of $\delta$.  Observe that similarly to the reservation vectors,  (a) the
entries $m_{si}$ in the source vector are multiples of $i\delta$ and (b) the entry
$m_{s0}$ is a multiple of $\delta$. In addition, (c) the entries in $m_s$ sum up to
at most $\lceil 1\rceil_\delta$ as only vertices of color $s$ contribute.
Observe that the source vector of a server $s$ just depends on the sizes of
the $s$-monochromatic pieces and on which of their vertices are committed;
\emph{it does not depend on how an algorithm assigns the pieces to servers}.

A vector $m$ is \emph{a potential source vector} if it fulfills properties
(a)-(c) without necessarily being the source vector for a particular server.
Similarly, a \emph{potential reservation vector} $r$ is a
vector where the $i$-th entry is a multiple of $i\delta$, the $0$-th entry a
multiple of $\delta$, and $r$ is $\gamma$-valid. Here, we say that 
$r$ is \emph{$\gamma$-valid} if $\|r\|_1\le 1+\gamma$. Note that there are
only $O(1)$ potential reservation or source vectors since they have only
$O(1/\delta)=O(1)$ entries (one per size class) and for each entry there are
only $O(1/\delta)=O(1)$ choices.

A \emph{configuration $(r,m)$} is a pair consisting of a potential
reservation vector $r$ and a potential source vector $m$. We
further call a configuration $(r,m)$ \emph{ordinary} if $r\ge m$ (i.e., $r_i\ge
m_i$ for each $i$) and otherwise we call it \emph{extraordinary}. The intuition
is that servers with ordinary configurations have enough reserved space such
that they can be assigned all of their monochromatic pieces. Next, note that as
there are only $O(1)$ potential source and reservation vectors, there are only
$O(1)$ configurations in total.
\begin{claim}
\label{onl:claim:numberOfConfigurations}
	There exist only $O(1)$ different configurations $(r,m)$.
\end{claim}

In the following, we assign configurations to servers and we will
call a server \emph{ordinary} if its assigned configuration is ordinary and
\emph{extraordinary} if its assigned configuration is extraordinary.

\medskip
We now define the ILP. Remember that the goal in this step is to
obtain a set of configurations, which we will then assign to the servers and which
will guide the assignment of the pieces to the servers. Thus, we introduce a
variable $x_{(r,m)}\in\mathbb{N}_0$ for each (ordinary or extraordinary)
configuration $(r,m)$.  After solving the ILP, our schedules will use exactly
$x_{(r,m)}$ servers with configuration $(r,m)$.
Furthermore, the objective function of the ILP is set such that the number of
extraordinary configurations is minimized.

The constraints of the ILP are picked as follows.  First, let $V_i$, $i>0$,
denote the total committed volume of all pieces in class $i$ and let $V_0$
denote the total uncommitted volume of all pieces. Note that the $V_i$ do not
depend on the schedule of the algorithm. Now we add a set of constraints, which
ensures that the configurations picked by the ILP reserve enough space such that
all pieces of class $i$ can be assigned to one of the servers.  Second, let
$Z_m$ denote the \emph{number} of servers with the potential source vector $m$ at this point in
time.
(Recall that the source vectors of the servers only depend on the current graph and the commitment decisions of the algorithm and \emph{not} on the
 algorithm's schedule.) We add a second set of constraints which ensures that for
each $m$, the ILP solution contains exactly $Z_m$ configurations with source
vector $m$.
Now the ILP is as follows.

\begin{equation*}
\begin{array}{llll}
\min\,\,
&\multicolumn{2}{l}{\sum\nolimits_{(r,m):\,r\not\ge m}x_{(r,m)}}& \\[2ex]
\text{s.t.}&\sum\nolimits_{(r,m)}x_{(r,m)}r_i/\delta   &\ge V_i/\delta&\,\,\text{for all $i$}\\[2ex]
&\sum\nolimits_r x_{(r,m)}                  &= Z_m&\,\,\text{for all $m$}
\end{array}
\end{equation*}
In the ILP we wrote $r_i/\delta$ and $V_i/\delta$ to ensure that ILP
only contains integral values. Further observe that the ILP has constant size
and can, hence, be solved in constant time: As there are only $O(1)$ different
configurations (Claim~\ref{onl:claim:numberOfConfigurations}), the ILP only has $O(1)$ variables. Also, since there are only
$O(1)$ size classes~$i$ and $O(1)$ source vectors~$m$, there are only $O(1)$
constraints.

Next, we show that an optimal ILP solution
serves as a lower bound on the cost paid by $\OPT$.
\begin{lemma}
\label{onl:lem:lowerILP}
	Suppose the objective function value of the ILP is $h$, then
	$\cost(\OPT)\ge (\gamma-\delta)h=\Omega(h)$.
\end{lemma}

\subsection{Schedules That Respect an ILP Solution}
\label{onl:sec:ilp-to-schedule}
Next, we describe the relationship of schedules and configurations. 
A \emph{schedule $S$} is an assignment of pieces to servers. The set of pieces
assigned to a particular server $s$ is called the \emph{schedule for $s$}.
A schedule for a server $s$ with source vector $m_s$ \emph{respects a reservation} $r$
if the following holds:
\begin{enumerate}[itemsep=0pt]
\item The committed volume of class $i$ pieces scheduled on $s$ is at most $r_{i}$.\label{onl:p:one}
\item The total uncommitted volume scheduled on $s$ is at most
$r_{0}+14\epsilon$.\label{onl:p:two}
\item If $r\ge m_s$ then all pieces that are monochromatic for $s$ are placed
on $s$.\label{onl:p:three}
\end{enumerate}
A schedule \emph{respects an ILP solution $x$} if there exists an assignment of
configurations to servers such that:
\begin{itemize}[itemsep=0pt,label=--]
  \item A server $s$ with source vector $m_s$ is assigned a configuration
  $(r,m)$ with $m=m_s$.
  \item A configuration $(r,m)$ is used exactly $x_{(r,m)}$ times.
  \item The schedule of each server respects the reservation of its assigned configuration.
\end{itemize}
Note that if all source vectors are different, then assigning the configurations
to the servers is clear (each server $s$ is assigned the configuration $(r,m)$ with
$m=m_s$). In the case that some source vectors appear in multiple
configurations, we describe a procedure for assigning the configurations to the
servers below.

The next lemma shows that servers respecting a reservation only
slightly exceed their capacities.
\begin{lemma}
\label{onl:lem:augmentation}
	If the schedule for a server $s$ respects a $\gamma$-valid reservation $r$, 
	then the total volume of all pieces assigned to $s$ is at most
	$1+\gamma+14\epsilon = 1+O(\varepsilon)$.
\end{lemma}

\subsection{How to Find Schedules}
In this section, we describe how to resolve the ILP and adjust the existing
schedule after a merge or commit-operation so that it respects the ILP solution,
 in particular,
that the schedule of every server respects the reservation of its assigned configuration, i.e., Properties 1-3 above. 
It is crucial that this step can be performed at a small cost. We present
different variants: In most situations, the algorithm uses a generic
variant that is based on sensitivity analysis of ILPs. However, in some
special cases (cases in which $\OPT$ might pay very little) using the generic
variant might be too expensive. Therefore, we develop special variants for
these cases that resolve the ILP and adjust the schedule at zero cost.

Before we describe our variants in detail, note that it is not clear how to
assign small pieces to servers based on the ILP solution. Hence, we define our
variants such that in the first phase they move some pieces around to
construct a respecting schedule but they ignore Property~\ref{onl:p:two} while doing
so, i.e., they only guarantee Property~\ref{onl:p:one} and Property~\ref{onl:p:three}.
After this (preliminary) schedule has been constructed, we run a
\emph{balancing procedure} (described below), which ensures that
Property~\ref{onl:p:two} holds. The balancing procedure only moves small pieces and
we show that its cost is at most the cost paid for the first phase. As it is relatively short, we describe it first.

\textbf{Balancing Procedure for Small Pieces.}
\def\slack{\operatorname{slack}}
We now describe our balancing procedure, which moves only small pieces and for
which we show that Property~\ref{onl:p:two} of respecting schedules is satisfied
after it finished. The balancing procedure is run after one of the variants of
the ILP solving is finished.

For a server $s$, let $v_u(s)$ denote the total uncommitted volume scheduled at
$s$. We define the slack of a server $s$ by $\slack(s):= r_{s0}-v_u(s)$.  Note
that because of the first constraint in the ILP with $i=0$, there is always a
server with non-negative slack.
Next, we equip every server $s$ with an \emph{eviction budget} $\mathit{budget}(s)$ that is initially $0$.
Now, any operation outside of the balancing procedure that decreases the slack
must increase the eviction budget by the same amount. Such an operation could,
e.g., be a piece $p$ that is moved to $s$ (which decreases the slack by $|p|_u$)
or a decrease in $r_{s0}$ when a new configuration is assigned to $s$. (Note that
increasing the eviction budget increases the cost for the operation performing
the increase; we will describe how we charge this cost later.) Intuitively it should
follow that $\mathit{budget}(s)$ roughly equals $-\slack(s)$ and indeed we can show through
a careful case analysis that $\mathit{budget}(s) \ge -\slack(s) - 2\epsilon$
(see Claim~\ref{onl:cla:budget}).

This eviction budget is used to pay for the cost of moving small pieces
away from $s$ when the balancing procedure is called. We say a small piece $p$ is
\emph{movable} if either its majority color has at most a $(1-2\epsilon)$-fraction of the volume of $p$ (the piece is far from monochromatic), or its
majority color corresponds to an extraordinary server.
The balancing procedure does the following for each server~$s$:
\begin{algorithmic}
\While{there is a movable piece $p$ on $s$ with $|p|<\mathit{budget}(s)$}
	\State Move $p$ to a server with currently non-negative slack.
	\State $\mathit{budget}(s) = \mathit{budget}(s)-|p|$
\EndWhile
\end{algorithmic}
The following lemma shows that after balancing procedure finished,
$\slack(s)\geq-14\varepsilon$.
This implies that when the balancing procedure finished,
Property~\ref{onl:p:two} holds since
$v_u(s) = r_{s0} - \slack(s) \leq r_{s0} + 14\varepsilon$.
\begin{lemma}
\label{onl:lem:slack}
	After the balancing procedure for a server $s$ finished, we have that
	$\slack(s)\geq -14\epsilon$.
\end{lemma}

\subsubsection{Overview of the Variants}
Next, we give the full details of the main algorithm for adjusting the schedules in
different cases.

\begin{description}[noitemsep]
	 \item[Step~I] \emph{Move small to large piece:}
	 Move the smaller piece $\ps$ to the server of the larger piece $\pl$.
     \item[Step~II] \emph{Merge pieces:}
		Merge $\ps$ and $\pl$ into $\pm$.
		Then adjust the schedule as follows:
          \begin{itemize}[topsep=0pt,label=--]
          \item If $\ps$ is small, then the ILP does not change and no adjustment is necessary.
          \item Else: if the merge is non-monochromatic \emph{or} $s$ is extraordinary,
					use the Generic Variant, otherwise, use Special Variant~A.
		  \item Run the rebalancing procedure.
          \end{itemize}
     \item[Step~III] \emph{Commit volume:} If $|\pm|\geq\varepsilon$, then
			\noindent
			\intextsep0pt
			\begin{algorithmic}
			\While{$|\pm|_u > 2\delta$:}
				\State Commit volume $\delta$ for $\pm$. 
				\StateN{If $\pm$ is non-monochromatic \emph{or} $s$ is extraordinary,
							use the Generic Variant,
						otherwise, use Special Variant~B.}
				\State{Run the rebalancing procedure.}
			\EndWhile
			\end{algorithmic}
\end{description}

\subsubsection{The Generic Variant}
\label{onl:sec:variantGeneric}
We now describe the Generic Variant of the schedule adjustment.  Suppose the set
of pieces~$\P$ changed into $\P'$ due to a merge or a commit-operation.

The algorithm always maintains for the current set $\P$ an optimum ILP solution. 
Let $x$ be the ILP solution for $\P$. When $\P$ changes, the algorithm runs the ILP to obtain the optimum
ILP solution $x'$ for $\P'$.

In the following, we first argue how to assign the configurations
from $x'$ to the servers and then we argue how we can transform a schedule $S$
(respecting $x$) into a schedule $S'$ (respecting $x'$) with little cost.

We first assign the configurations given by $x'$ to servers
by the following greedy process. A configuration
$(r,m)$ is \emph{free} if it has not yet been assigned to $x'_{(r,m)}$ servers. As long as
there is a free configuration $(r,m)$ and a server $s$
that had been assigned $(r,m)$ in schedule $S$, we assign $(r,m)$ to $s$. The
remaining configurations are assigned arbitrarily subject to the constraint
that a server $s$ with source vector $m_s$ obtains a configuration of the form
$(r,m_s)$ for some $r$.

\medskip
\noindent
Now that we have assigned the configurations to the servers, we still have to
ensure that the new schedule respects these new server configurations. We start
with some definitions.

First, let $\mathcal{A}$ be the set of servers for which the set of scheduled pieces
changed due to the merge- or commit-operation. For a merge-operation, these are the
servers that host one of the pieces $\ps,\pl$, or $\pm$, and for a
commit-operation, this
is the server that hosts the piece $\pm$ that executes the commit. Note that
$|\mathcal{A}|\leq 3$.
Second, let $\mathcal{B}$ be the set of servers that changed their source vector due to
the merge or commit-operation. Note that for a commit-operation
$|\mathcal{B}|$ could
be large, because the committed volume could contain many different colors and
for each corresponding server, the source vector could change by a reduction of
$m_{0}$.
Third, we let $\mathcal{C}$ be the set of servers that changed their assigned
configuration between $S$ and the current schedule~$S'$. Note that
$|\mathcal{C}|\le |\mathcal{B}|+\|x-x'\|_1$.

Observe that for servers $s\not\in\mathcal{A}\cup\mathcal{C}$ neither
their assigned configuration (since $s\not\in\mathcal{C}$) nor their set of
scheduled pieces (since $s\not\in\mathcal{A}$) has changed. Thus, these servers
already respect their configuration and, hence, we do not move any pieces for
these servers now.  For the servers in $\mathcal{A}\cup \mathcal{C}$ we do the
following:
\begin{enumerate}[itemsep=0pt]
	\item We mark all pieces currently scheduled on servers in $\mathcal{A}\cup
		\mathcal{C}$ as \emph{unassigned}.
	\item Every ordinary server in $\mathcal{A}\cup \mathcal{C}$
		moves all of its
		monochromatic pieces to itself. This guarantees Property~\ref{onl:p:three} of a
		respecting schedule. Note that this step may move pieces away from servers
		in $\overline{\mathcal{A} \cup \mathcal{C}}$.

	\item The remaining pieces are assigned in a first fit fashion. We say a server
		is \emph{free} for class $i>0$ if the \emph{committed volume} of class
		$i$ pieces already scheduled on it is (strictly) less than $r_i$. It is
		\emph{free} for class $0$ if the uncommitted volume scheduled on it is
		less than $r_0$. 
 
		To schedule an unassigned piece $p$ of class $i$, we determine
		a free server for class $i$ and schedule $p$ there. The first set of
		constraints in the ILP guarantees that we always find a free server.
\end{enumerate}
This scheduling will guarantee Property~\ref{onl:p:one} of a respecting schedule,
i.e., for all $i>0$ the volume of class $i$ pieces scheduled on a server $s$ is
at most $r_{si}$. This holds because $r_{si}$ is a multiple of $i\delta$. If we
decide to schedule a class $i$ piece on $s$ because a server is free for class
$i$ then it actually has space at least $i\delta$ remaining for this class.
Hence, we never overload class $i$, $i>0$.

\medskip
\noindent
In the following, we develop a bound on the cost of the above scheme. For the
analysis of our overall algorithm we use an involved amortization scheme.
Therefore, the cost that we analyze here is not the real cost that is incurred
by just moving pieces around but it is inflated in two ways:
\begin{enumerate}[label=(\Alph*)]
	\item If we move a piece $p$ to a server $s$, we increase the eviction budget of
	$s$ by $|p|_u$.\label{onl:evict}
	\item Whenever we change the configuration of a server from a ordinary to
	extraordinary, we experience an \emph{extra cost} of
	$4(1+\gamma)/\delta$.\label{onl:extra} This will be required later in
	Case~\ref{onl:case:ps-small} of the analysis.
\end{enumerate}
Observe that Cost Inflation~\ref{onl:evict} clearly only increases the cost by
a constant factor. Cost Inflation~\ref{onl:extra} will also only increase the
cost by a constant factor as the analysis below assumes constant cost for every
server that changes its configuration. Note that the Generic Variant is the
only variant for adjusting the schedule for which Inflation~\ref{onl:extra} has an
affect; the other variants do not move pieces around and do not generate any
new extraordinary configurations.

The following lemma provides the sensitivity analysis for the ILP. Its first
point essentially states that for adjusting the schedules, we need to pay cost
proportional to the number of servers that change their source configuration
from $\P$ to~$\P'$ plus the change in the ILP solutions. The second point then
bounds the change in the ILP solutions by the number of servers that change
their source vectors from $\P$ to~$\P'$.
\begin{lemma}
\label{onl:lem:adjustLemma}
	Suppose we are given a schedule $S$  that respects an ILP solution $x$ for a set
	of pieces $\mathcal{P}$. Let $\mathcal{P}'$ denote a set of pieces obtained
	from $\mathcal{P}$ by either a
	merge or a commit-operation, and let $\dif$ denote the number of servers that have a different
	source vector in $\mathcal{P}$ and $\mathcal{P}'$. Then:
	\begin{enumerate}
		\item\label{onl:lem:adjustSchedule}
			If $x'$ is an ILP solution for $\P'$,
			then we can transform $S$ into $S'$ with cost $O(1+\dif+\|x-x'\|_1)$.
		\item\label{onl:lem:adjustSolution}
			Then we can find an ILP solution $x'$ for $\P'$ with $\|x-x'\|_1= O(1+D)$.
		\item\label{onl:obs:mergeCheap}
			If the operation was a merge-operation, then $\dif\le 3$.
	\end{enumerate}
\end{lemma}

\subsubsection{Special Variant A: Monochromatic Merge}
\label{onl:sec:variantA}
Special Variant~A is used if we performed a monochromatic merge-operation of two
large pieces $\ps, \pl$ and if the server $s$ that holds the piece $\pl$ is
ordinary. Then $\OPT$ may not experience any cost. Therefore, we also want to
resolve the ILP and adjust the schedule $S$ with zero cost.

Since the merge is monochromatic, all of $\ps$, $\pl$ and $\pm$ are
monochromatic for $s$, and since $s$ has an ordinary configuration, $\ps$ and
$\pl$ are already scheduled at~$s$. Hence, the new piece $\pm$ (which is
generated at $\pl$'s server) is already located at the right server $s$.

We obtain our schedule $S'$ by deleting the assignments for $\ps$ and $\pl$
from $S$ and adding the location~$s$ for the new piece~$\pm$.
Now let $\is, \il$, and $\im$ denote the classes of pieces $\ps,\pl$, and
$\pm$, respectively (note that these classes are at least 1 as all pieces are large).
Then the new ILP can be obtained by only changing the configuration vector $m_s$
and setting
\begin{equation*}
	\begin{array}{lclll}
	m_{s{\is}}'    &:= &m_{s{\is}}    &- &|\ps|_c\\
	m_{s{\il}}'    &:= &m_{s{\il}}    &- &|\pl|_c\\
	m_{s{\im}}'    &:= &m_{s{\im}}    &+ &|\pm|_c\\
	\end{array},
	\text{~~~~}
	\begin{array}{lclll}
	Z_{m_s}'        &:= &Z_{m_s} &- &1\\
	Z_{m_s'}'         &:= &Z_{m_s'}  &+ &1\\
	\end{array}
\end{equation*}
and
\begin{equation*}
	\begin{array}{lclll}
	V_{\is}'        &:= &V_{\is}    &- &|\ps|_c   \\
	V_{\il}'        &:= &V_{\il} &- &|\pl|_c\\
	V_{\im}'        &:= &V_{\im}    &+ &|\pm|_c   \\
	\end{array}.
\end{equation*}

\medskip
To obtain a solution $x'$ to this new ILP, we change the
reservation vector for the server $s$ as follows.
\begin{equation*}
\begin{array}{lclll}
r_{s\is}'    &:= &r_{s\is}    &- &|\ps|_c   \\
r_{s\il}'    &:= &r_{s\il}    &- &|\pl|_c\\
r_{s\im}'    &:= &r_{s\im}    &+ &|\pm|_c   \\
\end{array}.
\end{equation*}
This does not change the $\|\cdot\|_1$-norm of the vector $r$ because 
$r_{\is}\ge m_{\is}\ge |\ps|_c$ (this follows from the definition of $m_{\is}$ and
the fact that $r_s\ge m_s$ holds) and because
$|\ps|_c+|\pl|_c = |\pm|_c$. 
We obtain the ILP solution $x'$ by setting
\begin{equation*}
	x_{(r_s,m_s)}'   := x_{(r_s,m_s)}   - 1
	\text{\,\,\,\,\,\,and\,\,\,\,\,\,}
	x_{(r_s',m_s')}' := x_{(r_s',m_s')} + 1.
\end{equation*}
Note that $r_s\ge m_s$ implies $r_s'\ge m_s'$. Hence, our new ILP solution does not
increase the objective function value of the ILP (i.e., the number of
extraordinary configurations).  In Lemma~\ref{onl:lem:ilp-solution-still-optimal} in
Section~\ref{onl:sec:ilp-solution-still-optimal} we show that merging two large monochromatic pieces of a server
cannot decrease the objective function value of the ILP. Therefore, the new
ILP solution $x'$, which has the same objective function value as~$x$, is
optimal.

Finally, observe that we only changed the configuration of server $s$ and that
we did not move any pieces.
Hence, we can transform $\mathcal{P}$, $x$ and $S$ into
$\mathcal{P}'$, $x'$ and $S'$ with zero cost.

\subsubsection{Special Variant B: Monochromatic Commit}
\label{onl:sec:variantB}
Suppose we perform a commit-operation for a monochromatic piece $\pm$ that is
located at an ordinary server~$s$. Then $\OPT$ may not experience any cost.
Therefore, we present a special variant for adjusting the schedule that also
induces no cost.  We perform a routine similar to Special Variant~A and provide
the details in Appendix~\ref{app:onl:sec:variantB}.

\def\bold #1{{\bfseries\mathversion{bold}#1}}
\makeatletter
\newcounter{step}\def\thestep{\Roman{step}}
\newcounter{case}\def\thecase{\alph{case}}
\def\step #1{\setcounter{case}{0}\stepcounter{step}\bold{Step~\thestep: #1.}}
\def\case
#1{\medskip\noindent\setlength\leftmargin{1cm}\stepcounter{case}\bold{Case
    (\thestep\thecase) #1.}%
\edef\@currentlabel{\thestep\thecase}%
}
\makeatother

\def\NM{\mathrm{NM}}
\section{Analysis}
\label{onl:sec:analysis}
We first give a high level overview of the analysis.
Let $\mathcal{P}^*$
denote the final set of pieces. A simple lower bound on the cost of $\OPT$ is
as follows. Let $\NM$ denote the set of vertices that do not have the majority
color within their piece in $\mathcal{P}^*$. Then $\cost(\OPT)\ge\frac{1}{k}|\NM|$,
because each vertex has volume $\frac{1}{k}$ and for each piece in
$\mathcal{P}^*$, $\OPT$ has to move all vertices apart from vertices of a
single color. Hence, the total volume of pieces moved by $\OPT$ is at least
$\frac{1}{k}|\NM|$. 

We want to exploit this lower bound by a charging argument. The general idea is
that whenever our online algorithm experiences some cost $C$, we \emph{charge}
this cost to vertices whose color does not match the majority color of their
piece. If the total charge made to each such vertex $v$ is at most
$\alpha\cdot\size(v)$, then the cost of the online algorithm is at most
$\alpha\cdot\cost(\OPT)$. When we charge cost to vertices, we will refer to this
as \emph{vertex charges}.

The difficulty with this approach is that at the time of the charge, we do not
know whether a vertex will have the majority color of its piece in the end.
Therefore, we proceed as follows.  Suppose we have a subset $S$ of vertices in a
piece $p$ and a subset $Q\subseteq S$ does not have the current majority color
of $S$ . Then \emph{regardless of the final majority color of $p$}, a total
volume of $\size(Q)$ of vertices in $S$ will not have this color in the end.
Hence, when we distribute a charge of $C$ evenly among the vertices of~$S$, a
charge of $\size(Q)\cdot C/\size(S)$ goes to vertices that do not have the final
majority color.  We call this portion of the  charge \emph{successful}.

The following lemma shows that to obtain algorithms competitive to $\OPT$, it
suffices if we bound the successful and the total vertex charges.
\def\chargesucc{\mathit{charge}_{\operatorname{succ}}}
\def\chargemax{\mathit{charge}_{\operatorname{max}}}
\begin{lemma}
\label{onl:lem:lowerVertex}
	Suppose the total successful charge is at least
	$\chargesucc$ while the maximum (successful and unsuccessful)
	charge to a vertex is at most $\chargemax$. Then
	$\cost(\OPT)\ge\frac{1}{k}\chargesucc/\chargemax$.
\end{lemma}
\begin{proof}
	Note that successful charge only goes to vertices in $\NM$. Hence,
	$|\NM|\ge\chargesucc/\chargemax$, and, therefore, we obtain that
	$\cost(\OPT)\ge\frac{1}{k}|\NM|\ge\frac{1}{k}\chargesucc/\chargemax$.
\end{proof}

Another lower bound that we use is due to Lemma~\ref{onl:lem:lowerILP}. Let $\hmax$
denote the maximum objective value obtained when solving different
ILP instances during the algorithm. From time to time, when vertex charges are
not appropriate, we perform \emph{extraordinary charges} or just \emph{extra
  charges}. In the end, we compare the total extra charge to $\hmax$. We stress
that we only perform extra charges when extraordinary configurations are
involved. This means if $\hmax=0$ we never perform extra charges, as otherwise,
it would be difficult to obtain a good competitive ratio.

In the following analysis, we go through the different steps of the algorithm.
For every step, we charge the cost either by a vertex charge or by an extra
charge. If we apply a vertex charge, we argue that (1)~enough of the applied
charge is successful and (2)~the charge can accumulate to not too much
at every vertex.
For extra charges, we require a more global argument and we will derive a bound
on the total extra charge in terms of $\hmax$ in Section~\ref{onl:sec:extra}.

\subsection{Analysis Details}
When merging a piece $\pl$ and $\ps$ with $|\ps| \le |\pl|$ we proceed in
several steps.

\medskip
\noindent
\step{Small to Large}
In this first step, we move the vertices of $\ps$ to the server of $\pl$.
If $\ps$ and $\pl$ are on different servers we experience a cost of $|\ps|$.
Also, we have to increase the eviction budget of the server that holds
piece $\pl$ (if $\ps$ is a small piece). 
The cost for this step is $0$ if $\ps$ and $\pl$ are on the same server and,
otherwise, it is at most $2|\ps|$.
We charge the cost as follows.

\case{Merge is monochromatic}\label{onl:case:mono-mergeI}
If $\ps,\pl$, and $\pm$ are mono\-chromatic for the same server $s$ we only
experience cost if $s$ is extraordinary because otherwise $\ps$ and $\pl$ are
located at $s$. We make an extra charge for this cost.

\case{Merge is not monochromatic}\label{onl:case:non-mono-mergeI}
We make the following vertex charges:
\begin{itemize}
	\item \TypeI charge: We charge $\frac{2}{\delta} \cdot \frac{|\ps|}{|\pm|}\cdot\size(v)$ to every vertex in $\pm$. 
	\item \TypeII charge: We charge $\frac{2}{\delta} \cdot\size(v)$ to every vertex in $\ps$.
\end{itemize}
Claim~\ref{onl:cla:chargeIandII} below shows that the \TypeI and \TypeII charge
at a vertex can accumulate to at most $O(\log k)$.
In the following, we argue that at least a charge of $2|\ps|$ is successful.
 We distinguish several cases.
\begin{itemize}
	\item If either $\pl$ or $\pm$ is not monochromatic, we know that at least a
	volume of $\delta$ (if the non-monochromatic piece is large) or a volume of
	$\epsilon|\pl|$ of vertices does not have the majority color. Hence, we get that
	at least $\min\{\delta,\epsilon|\pl|\}\frac{2|\ps|}{\delta|\pm|}\ge 2|\ps|$ of the \TypeI
	charge is successful. The inequality uses $|\pm|\le 1$,
	$|\pl|\ge\frac{1}{2}|\pm|$, and $\delta\le \epsilon^2\le\epsilon/2$.

	\item If $\ps$ is not monochromatic then at least $\delta|\ps|$ volume in $\ps$
	has not the majority color. This gives a successful charge of at least
	$\delta|\ps|\cdot\frac{2}{\delta}\ge2|\ps|$.

	\item Finally suppose that $\ps$ and $\pl$ are monochromatic for different
	colors $C_s$ and $C_\ell$, respectively. If in the end $C_s$ is not the majority
	color of the final piece then we have a successful charge of at least
	$(1-\epsilon)|\ps|\cdot 2/\delta\ge2|\ps|$ from the \TypeII charge. Otherwise,
	$C_\ell$ is not the majority color and we obtain a successful charge of
	$(1-\epsilon)|\pl|\cdot \frac{2|\ps|}{\delta|\pm|}\ge2|\ps|$.
\end{itemize}

\begin{claim}
\label{onl:cla:chargeIandII}
	The combined \TypeI and \TypeII charge that can accumulate at a vertex $v$ is
	at most $O(\log k \cdot \size(v)/\delta)$.
\end{claim}

\medskip
\noindent
\step{Resolve ILP and Adjust Schedule}
In this step, we merge the pieces $\ps$ and $\pl$ into $\pm$ and run the
subprocedure for adjusting the schedule, which finds a new optimum solution to
the ILP and finds a schedule respecting the ILP solution. Due to
Lemma~\ref{onl:lem:adjustLemma}
this incurs at most constant cost. In the following, we distinguish several
cases. For some cases, the bound of
Lemma~\ref{onl:lem:adjustLemma} is sufficient and we only have
to show how to properly charge the cost. For other cases, we give a
better bound than the general statement of
Lemma~\ref{onl:lem:adjustLemma}.
In the following, $s$ denotes the server where the merged piece $\pm$ is
located now (and where $\pl$ was located before).

\case{$\ps$ small}\label{onl:case:ps-small}
In this case, the input to the ILP did not change. This holds because no volume
was committed and no uncommitted volume changed between classes. Therefore we
do not experience any cost for resolving the ILP.

However, it may happen that $\pl$ was not monochromatic but the merged
piece $\pm$ is. Note that this can only happen if $\pl$ is also small (since
large pieces cannot transition from non-monochromatic to monochromatic).
Suppose $\pm$ is monochromatic for a server $s'\neq s$, and this server has an
ordinary configuration. Then we have to move $\pm$ to $s'$ for the new
schedule to respect the configuration of $s'$. We incur a cost of 
$|\pm|+|\pm|_u\le 2|\pm|$, where $|\pm|_u$ is required to increase the eviction
budget at $s'$. (Indeed, moving $\pm$ to $s'$ might cause the rebalancing procedure
to move pieces away from $s'$.) We charge $4/\delta\cdot\size(v)$ to every vertex in
$\pm$. We call this charge a \TypeIII charge.

How much of the charge is successful? Observe that $\pl$ was not monochromatic
for $s'$ before the merge as otherwise it would have been located at $s'$. This
means vertices with volume at least $\delta|\pl| \ge \delta|\pm|/2$ in $\pm$ have a
color different from $s'$ (the majority color in $\pm$). This means we get a
successful charge of at least $\delta|\pm|/2\cdot 4/\delta=2|\pm|$, as desired.

\bigskip To obtain a good bound on the total \TypeIII charge
accumulating at a vertex $v$ we have to add a little tweak. Whenever a
server $s$ switches its configuration from ordinary to extraordinary, we
cancel the most recent \TypeIII charge operation for all vertices
currently scheduled on $s$.

This negative charge is accounted for in the \emph{extra cost} that we
pay when switching the configuration of a server from ordinary to extraordinary.
Recall that in Cost Inflation~\ref{onl:extra}, we said
that we experience an extra cost of $4(1+\gamma)/\delta$ whenever we
switch the configuration of a server $s$ from ordinary to extraordinary.
This cost is used to cancel the most recent Type~III charge for all
pieces currently scheduled on~$s$.

\begin{lemma}
\label{onl:lem:typeIII}
	Suppose a vertex $v$ experiences a positive \TypeIII charge at time $t$ that is
	not canceled. Let $t'$ denote the time step of the next \TypeIII charge for
	vertex $v$, and let $p$ and $p'$ denote the pieces that contain $v$ at
	times $t$ and $t'$, respectively.
	Then $|p'|\ge (1+\epsilon)|p|$.
\end{lemma}

\begin{corollary}
\label{onl:cor:chargeIII}
	The total \TypeIII charge that can accumulate at a vertex is only $O(\log k\cdot\size(v))$.
\end{corollary}

\def\CtypeIV{C_{\mathrm{IV}}}
\case{$\ps$ large, merge not monochromatic}
We resolve the ILP and adjust the schedule $S$. According to
Item~\ref{onl:lem:adjustSchedule} and Item~\ref{onl:obs:mergeCheap} 
of Lemma~\ref{onl:lem:adjustLemma} this
incurs constant cost. Let $\CtypeIV$ denote the bound on this cost. We perform a vertex
charge of $\CtypeIV/\delta\cdot\size(v)$ for every vertex in $\pm$. We call this charge a \TypeIV
charge. In the following we argue that
at least a charge of $\CtypeIV$ is successful. We distinguish two cases.

If one of the pieces $\ps,\pl$, or $\pm$ is not monochromatic we know
that at least vertices of volume $\delta$ in the piece do not have the
majority color. Hence, we get that at least $\CtypeIV/\delta\cdot\delta\ge\CtypeIV$ of the \TypeIV
charge is successful.

Now, suppose that $\ps$ is monochromatic for server $s$ and $\pm$ is
monochromatic for a different server $s'$. Regardless of which color is the
majority color in the end, there will be vertices of volume at least
$(1-\epsilon)|\ps|$ that will not have this majority color. Hence, we obtain
a successful charge of at least $(1-\epsilon)|\ps|\cdot
\CtypeIV/\delta\ge(1-\epsilon)\epsilon\cdot\CtypeIV/\delta\ge\CtypeIV$, where the first step uses that $\ps$
is large and the second that $\delta\le\epsilon^2\le(1-\epsilon)\epsilon$,
which holds because $\epsilon\le1/4$.

\begin{claim}
\label{onl:cla:chargeIV}
	A vertex $v$ can accumulate a total \TypeIV charge of at most
	$\CtypeIV/\delta\cdot\size(v)$.
\end{claim}

\case{$\ps$ large, merge monochromatic, $s$ extraordinary}\label{onl:case:mono-mergeII}
In this case, we also
resolve the ILP and adjust the schedule, which according to
Item~\ref{onl:lem:adjustSchedule} and Item~\ref{onl:obs:mergeCheap} of Lemma~\ref{onl:lem:adjustLemma} incurs
constant cost. Let $C$ denote this cost. We make an extra charge of $C$.
Observe that $C=O(|\ps|)$ because $\ps$ is a large piece. This will be
important when we derive a bound on the total extra charge.

\case{$\ps$ large, merge monochromatic, $s$ ordinary}\label{onl:case:ordinary-mono-merge}
Suppose that the server $s$ has an ordinary configuration. In this case we do
not want to have any cost, because we cannot perform an extra charge as no
extraordinary configurations are involved and we cannot charge against the
vertices of $\pm$ as the piece is monochromatic.
We use Special Variant~A for adjusting the schedule. This induces zero cost.

\def\CtypeV{C_{\mathrm{V}}}

\bigskip
\noindent
\step{Commit-operation} We analyze the commit-operation. We will call a
commit-operation monochromatic if it is performed on a monochromatic piece and,
otherwise, we call it non-monochromatic.

\case{$\pm$ not monochromatic, $s$ ordinary}
The commit-operation may change the source vector of several servers. Let
$\dif$ denote the number of servers that changed their source vector. The
cost for handling the commit-operation is at most $O(1+\dif)$ according to
Lemma~\ref{onl:lem:adjustLemma}. Let $\CtypeV$
denote the hidden constant, i.e., the cost is
at most $\CtypeV(1+\dif)$. We split this cost into two parts: $\CtypeV$ is the
\emph{fixed cost} and $\CtypeV\dif$ is the \emph{variable cost} of the commit. 

We charge $3\CtypeV/\delta\cdot\size(v)$ to every vertex $v$ in $\pm$. We call this
charge a \TypeV charge. In $\pm$ at
least vertices of volume $\delta$ have not the majority color because $\pm$ is
not monochromatic. Therefore we get a successful charge of
$3\CtypeV/\delta\cdot\delta=3\CtypeV$.

Clearly, the charge is sufficient for the fixed cost. However, the remaining
successful charge of $2\CtypeV$ may not be sufficient for the variable cost. 
In the following, we argue that the total remaining successful charge
that is performed \emph{for all} non-monochromatic commits is enough to cover
the variable cost for these commits.

\def\Im{I_{\mathrm{m}}}
\def\Inm{I_{\mathrm{nm}}}
\def\Nm{X_{\mathrm{m}}}
\def\Nnm{X_{\mathrm{nm}}}

\begin{lemma}
\label{onl:lem:monocommit}
	Let $\Nnm(s)$ denote the number of times that a non-monochromatic commit
	causes a change in the source vector of $s$. Then the variable cost for all
	non-monochromatic commits is at most $\sum_s\CtypeV\Nnm(s)\le 2\CtypeV N$,
	where $N$ denotes the total number of non-monochromatic commits.
\end{lemma}

Observe that the total remaining charge for the non-monochromatic commits is $2\CtypeV N$ (a
charge of $2\CtypeV$ for every commit). Hence, the previous lemma implies that this
remaining charge is sufficient for the variable cost of all non-monochromatic commits.

\begin{claim}
\label{onl:cla:chargeV}
	The \TypeV charge at a vertex $v$ can accumulate to at most
	$3\CtypeV/\delta^2\cdot\size(v)$.
\end{claim}

\case{$\pm$ monochromatic, $s$ ordinary}\label{onl:case:ordinary-mono-commit}
Suppose we perform a commit-operation for the piece $\pm$.
Here we use Special Variant~B for resolving the ILP and adjusting
the schedule. This incurs zero cost.

\def\Cmono{C_1}
\case{$\pm$ monochromatic, $s$ extraordinary}\label{onl:case:mono-commit}
We resolve the ILP and adjust the schedule. The cost for this is $O(1)$,
since we can use Item~\ref{onl:lem:adjustSchedule} of Lemma~\ref{onl:lem:adjustLemma}
with $D=1$, because $\pm$ is monochromatic and thus we only commit volume of color $s$.
Let $\Cmono$ denote the upper bound for this cost.
We perform an extra charge of $\Cmono$.
Since the committed volume has only color $s$, the total number of monochromatic commits
for a specific server $s$ is at most $1/\delta=O(1)$ because each commit
increases the committed volume of color $s$ by $\delta$. Consequently, the total extra charge
that we perform for monochromatic commits of a specific server $s$ is at most $\Cmono/\delta$. To
simplify the analysis of the total extra charge in Section~\ref{onl:sec:extra} we combine
all these extra charges into one extra charge of $\Cmono/\delta$ that is performed
whenever the server $s$ switches its state from ordinary to extraordinary
\emph{for the first time}.

\subsubsection{Analysis of Extra Charges}
\label{onl:sec:extra}
In this section we derive a bound on the total extra charge generated by our
charging scheme. Let us first recap when we perform extra charges:
\begin{enumerate}[label=(\Roman*)]
	\item\label{onl:extra:A} During the merge-operation we perform an extra charge of $O(|\ps|)$
	in Case~\ref{onl:case:mono-mergeI} and
	Case~\ref{onl:case:mono-mergeII}, when the merge-operation is monochromatic
	for server $s$ and $s$ has an extraordinary configuration.

	We stress the fact that whether a merge is monochromatic only depends on the
	sequence of merges and not on the way that pieces are scheduled by our
	algorithm.
	\item\label{onl:extra:B} Whenever a server changes its configuration from ordinary to extraordinary
	\emph{for the first time}, we generate an extra charge of $\Cmono/\delta=O(1)$
	to take care of the cost of monochromatic commits
	(Case~\ref{onl:case:mono-commit}).
\end{enumerate}

\medskip\noindent
Now let $\hmax$ denote the maximum number of extraordinary
configurations that are used throughout the algorithm. Clearly, if $\hmax=0$
there is never any extraordinary configuration and the extra charge will be
zero. If $\hmax\geq 1$, we show that the previously described deterministic
online algorithm guarantees an extra charge of at most $O(\ell \log k)$.
\begin{lemma}
\label{onl:lem:totalExtraCharge}
	If $\hmax=0$, there is no extra charge.
	If $\hmax\geq1$, the total extra charge is $O(\ell\log k)$.
\end{lemma}

Next, we show that the maximum vertex
charge (successful or unsuccessful) is $O(\lg k \cdot \size(v))$.
\begin{lemma}
\label{onl:lem:vertexcharge}
	The maximum vertex charge $\chargemax$ (successful or unsuccessful) that a
	vertex $v$ can receive is at most $O(\log k\cdot\size(v))$.
\end{lemma}

Combining Lemma~\ref{onl:lem:totalExtraCharge} and Lemma~\ref{onl:lem:lowerILP}
for extra charges and our arguments about vertex charges with
Lemma~\ref{onl:lem:lowerVertex}, we obtain the following theorem.

\def\ALG{\operatorname{ALG}}
\begin{theorem}
\label{onl:thm:det}
	There exists a deterministic online algorithm with competitive ratio $O(\ell\log k)$.
\end{theorem}

Note that we obtain an even stronger result if $\hmax=0$: the cost is at most $O(\log k)\cdot\cost(\OPT)$
because of the bound on the total vertex charge (and the fact that we do not
have extra charges). Otherwise~($\hmax>0$), the total extra charge is at most
$O(\ell\log k)$, which means that we are $O((\ell\log k)/\hmax)$-competitive.
So the worst-case competitive ratio occurs when $\hmax=1$.

The constant hidden in the $O(\cdot)$-notation in the theorem is
$(1/\varepsilon)^{O(1/\varepsilon^4)}$.
The exponential dependency on $1/\varepsilon$ is caused by the ILP sensitivity
analysis in 
Lemma~\ref{onl:lem:adjustLemma}. In particular, the hidden constants in 
Items~\ref{onl:lem:adjustSchedule} and~\ref{onl:lem:adjustSolution}
of the lemma are $(1/\varepsilon)^{O(1/\varepsilon^4)}$, since the ILP has
one variable for each potential configuration and the number of such
configurations in Claim~\ref{onl:claim:numberOfConfigurations} is
$(1/\varepsilon)^{O(1/\varepsilon^2)}$.
All other steps of the analysis only add factors $\poly(1/\varepsilon)$.

Next, consider the case $\varepsilon>1$. Then the servers can store vertices of
volume $2+\varepsilon'$ and the above algorithm is $O(\lg k)$-competitive:
Indeed, in this case, servers can always store all of their monochromatic pieces
(of total volume at most $1$) and never become extraordinary; thus, we are in
the setting with $\hmax=0$ above. Furthermore, the algorithm above never assigns
pieces of volume more than $1+\varepsilon'$ to each server. Together, this
bounds the total load of each server to $2+\varepsilon'$ and we obtain the
following theorem.
\begin{theorem}
\label{onl:thm:det-large-epsilon}
	If $\varepsilon>1$, there exists a deterministic online algorithm with
	competitive ratio $O(\log k)$.
\end{theorem}

\section{Randomized Algorithm}
\label{onl:sec:randomized-algorithm}
How can randomization help to improve on the competitive ratio? For this
observe that the cost that we charge to vertices is at most
$O(\log k\cdot \cost(\OPT))$. Hence, the critical part is the cost for which we
perform extra charges, which can be as large as $\Omega(\ell\log k)$ according
to Theorem~\ref{onl:thm:lb-det}. A rough sketch of a (simplified)
lower bound is as follows. We generate a scenario where initially all servers
have the same source vector but some server needs to schedule its
source-pieces on
different servers (as, otherwise, we could not fulfill all constraints).

In this situation, an adversary can issue merge requests for all vertices that
originated at the server $s$ that currently has its source-pieces distributed
among several servers. Then the online algorithm incurs constant cost to
reassemble these pieces on one server, and, in addition, has to split the
source-pieces of another server between at least two servers.
Repeating this for $\ell-1$ steps gives a cost of $\Omega(\ell)$ to the online
algorithm while an optimum algorithm just incurs constant cost. 

The key insight for randomized algorithms is that the above scenario cannot
happen if we randomize the decision of which server distributes its
source-pieces among several servers. The online problem then turns into a
paging problem and we use results from online paging to derive our bounds.

\subsection{Augmented ILP}
\label{onl:sec:augmented-ilp}
Let $M$ denote the set of all potential source vectors. We introduce a partial
ordering on $M$ as follows. We say $m \ge_p m'$ if any prefix-sum of $m$ is
at least as large as the corresponding prefix-sum for $m'$. Formally,
\begin{equation*}
m\ge_p m' ~~~\text{$\Longleftrightarrow$}~~~ \text{$\forall i$:\,\,}{\textstyle\sum_{j=0}^im_j\ge\sum_{j=0}^im'_j}\enspace.
\end{equation*}
Observe that $m\ge m'$ implies $m\ge_p m'$. We adapt the ILP by adding a
cost-vector $c$ that favors large source vectors w.r.t.\ $\ge_p$. This
means as a first objective the ILP tries to minimize the number of
extraordinary configurations as before but as a tie-breaker it favors
extraordinary configurations with large source vectors. For this we assign
unique ids from $1,\dots, |M|$ to the source vectors s.t.\
$m_1\ge_p m_2 \implies \operatorname{id}(m_1)\le \operatorname{id}(m_2)$.
Then we define the cost-vector $c$ by setting
\begin{equation}
\label{onl:eq:cost-vector}
c_{(r,m)}:=\left\{\begin{array}{l@{~~}l}0 & r\ge m\\1+\lambda\operatorname{id}(m)&\text{otherwise}\end{array}\right.\enspace,
\end{equation}
for $\lambda = 1/(|M|^2 \cdot \ell)$. Given the cost-vector $c$, we set the
objective function of our new ILP to $\sum_{(r,m)}c_{(r,m)}x_{(r,m)}$.
The choice of $\lambda$ together with
$\|x\|_1 = \ell$ imply that
$\sum_{(r,m) : r \not\geq m} \lambda \operatorname{id}(m) x_{(r,m)} \leq \lambda
\cdot |M| \ell = 1/|M| < 1$. 
Thus, the ILP still minimizes 
the number of
extraordinary servers.

Note that the sensitivity
analysis for the ILP still holds
(Theorem~\ref{onl:thm:sensitivity} is independent of the cost vector and also
Lemmas~\ref{onl:lem:ilp-solution-still-optimal} and~\ref{onl:lem:ilp-solution-still-optimal-commit}
hold for the cost vector defined above).
This means if we have a constant change in the RHS~vector of the ILP, we can
adjust the ILP solution and the schedule at the cost stated in
Lemma~\ref{onl:lem:adjustLemma}.
Similarly, when we manually adjust the ILP solution
(Case~\ref{onl:case:ordinary-mono-merge} and Case~\ref{onl:case:ordinary-mono-commit}),
we do not increase the cost because only the configuration of a single server
$s$ changes and this server keeps its ordinary configuration, i.e., it does not
contribute to the objective function of the ILP.

A crucial property of the partial order $\geq_p$ is that source vectors of
servers are monotonically decreasing w.r.t.~$\geq_p$ as time progresses and
as more merge-op\-er\-a\-tions are processed.
\begin{observation}
\label{onl:obs:monotone}
	Let $m_s(t)$ denote the source vector of some server $s$ after some timestep
	$t$ of the algorithm. Then $t_1\le t_2$ implies $m_s(t_1)\ge_p m_s(t_2)$, i.e.,
	the source vector of a particular server is monotonically decreasing w.r.t.\
	$\ge_p$.
\end{observation}

\subsection{Marking Scheme}
\label{onl:sec:marking-scheme}
The total extra charge that is generated by our algorithm is determined by
how we assign extraordinary configurations to servers. We use a marking
scheme to decide which servers \emph{may} receive an extraordinary
configuration. 
Formally, a (randomized) \emph{marking scheme} dynamically partitions the
servers into \emph{marked} and \emph{unmarked} servers and satisfies the following
properties:
\def\cost{\operatorname{cost}}
\begin{itemize}
	\item Initially, i.e., before the start of the algorithm, all servers are
		unmarked. 
	\item Let $h_m$ denote the number of servers with source vector $m$ that
		are assigned an extraordinary configuration by the ILP, i.e.,
		$h_m=\sum_{(r,m):r\not\ge m}x_{(r,m)}$. The marking scheme has to mark
		at least $h_m$ servers with source vector $m$.
\end{itemize}
The cost $\cost(\mathcal{M})$ of a marking scheme $\mathcal{M}$ is defined as
follows:
\begin{itemize}
	\item Switching the state of a server from marked to unmarked or vice versa
	induces a cost of $1$.
	\item If a marked server experiences a monochromatic merge, the cost increases
	by $|\ps|$, where $\ps$ is the smaller piece involved in the merge-operation.
\end{itemize}
Suppose for a moment that the marked servers always are exactly the servers that
are assigned an extraordinary configuration. Then the above cost is clearly an
upper bound on the total extra charge as define in Section~\ref{onl:sec:extra} (up
to constant factors). This is because the marking scheme pays whenever switching
between marked and unmarked, while in our analysis we only make one extra charge
of constant cost when a server switches to an extraordinary configuration for
the first time.

In the following, we enforce the condition that a server only has an
extraordinary configuration if it is marked by the marking scheme. However,
the marking scheme could mark additional servers that are not extraordinary.
Thus, by enforcing this condition our algorithm incurs additional cost.
Suppose, e.g., that the marking scheme decides to unmark a server $s$ that is
currently marked and has been assigned an extraordinary configuration.
Then we have to switch the (extraordinary) configuration $(r,m_s)$
assigned to $s$ with an ordinary configuration $(r',m_s)$ that currently is
assigned to a different marked server $s'$. Note that we always find such a
server because there exist at least $h_{m_s}$ marked servers with
source vector $m_s$. The switch can then be performed at constant cost.
We make an additional extra charge for this increased cost of our algorithm.
Note that the marking scheme accounts for this additional cost as it incurs
cost whenever the state of a server changes. Therefore, the cost of the marking
scheme can indeed serve as an upper bound on the total extra charge (including
the additional extra charge). This gives the following observation.
\begin{observation}
\label{onl:obs:markingscheme}
	Let $\mathcal{M}$ be a marking scheme. The total extra charge is
	at most $O(\cost(\mathcal{M}))$.
\end{observation}

Next, we construct a marking scheme with small cost. For simplicity of
exposition we assume that we know $\hmax$, the maximum number of extraordinary
configurations that will be used throughout the algorithm, in advance. We
describe in Appendix~\ref{onl:sec:marking-doubling} how to adjust the scheme to work
without this assumption by using a simple doubling trick (i.e., make a guess for
$\hmax$ and increase the guess by a factor of 2 if it turns out to be wrong).

We will use results from a slight variant of online
paging~\cite{finely-competitive-paging}.  In this problem, a sequence of page
requests has to be served with a cache of size $z$. A request $(p,w)$ consists
of a page $p$ from a set of $\ell\ge z$ pages together with a weight
$w\le1$.\footnote{Note that our problem definition slightly differs from the
	model analyzed by Blum et al.~\cite{finely-competitive-paging}, which has
	$w=1$ for every request.  However, it is straightforward to show that the
	results of \cite{finely-competitive-paging} carry over to our model.}
If the requested page is in the cache, the cost for an algorithm serving the
request sequence is $0$.  Otherwise, an online algorithm experiences a cost of
$w$.  It can then decide to put the page into the cache (usually triggering the
eviction of another page) at an additional cost of $1$.

The cost metric for the optimal offline algorithm is different and provides an
advantage to the offline algorithm. If the offline algorithm does not have $p$
in its cache, it pays a cost of $w/r$, where $r\geq 1$ being a parameter of the
model, and then it can decide to put $p$ into its cache at an additional cost of
$1$.  In~\cite{finely-competitive-paging}, the authors show how to obtain a
competitive ratio of $O(r+\log z)$ in this model.

\textbf{The Paging Problems.}
Let $M$ denote the set of potential source vectors and recall that $|M|=O(1)$.
We introduce $|M|$ different paging problems, one for every potential source
vector $m\in M$.

Fix a potential source vector $m$. Let $S_m$ denote the set of servers that have a
source vector $m'\ge_p m$.
Essentially, we simulate a paging algorithm on the set 
$S_m$ (i.e., servers correspond to pages) with a cache of size $|S_m|-\hmax$
and parameter $r=\log k$.

Note that a server may leave the set $S_m$, but it is not possible for a server
to enter this set because the source vector $m_s$ of a server is non-increasing
w.r.t.\ $\ge_p$ (Observation~\ref{onl:obs:monotone}). The fact that servers may
leave $S_m$ is problematic for setting up our paging problem because this would
correspond to decreasing the cache size, which is usually not possible.
Therefore, we define the paging problem on the set of \emph{all servers} and we
set the cache size to $\ell-\hmax$, but we make sure that servers/pages not in
$S_m$ are always in the cache. This effectively reduces the set of pages to
$S_m$ and the cache size to $|S_m|-\hmax$.

We construct the request sequence of the paging problem for $S_m$ as follows.
A monochromatic merge for a server $s\in S_m$ is translated into a
page request for page $s$ with weight $|\ps|$, where $\ps$ is the smaller piece
that participates in the merge-operation. Following such a merge request, we
issue a page request (with weight~1) for every page/server not in $S_m$. This
makes sure that an optimum solution keeps all these pages in the cache at all
times, thus reducing the effective cache-size to $|S_m|-\hmax$.
The request sequence stops when $|S_m|=\hmax$.

\textbf{The Marking Scheme.}
We obtain a marking scheme from all the different paging algorithms as follows.
A server with source vector $m$ is marked if it is \emph{not} in the cache
for the paging problem on set $S_m$, or if $|S_m|\le\hmax$. The following lemma
shows that this gives a valid marking scheme.
\begin{lemma}
\label{onl:lem:marking-scheme}
	The marking scheme marks at least $h_m$ servers with source vector $m$.
\end{lemma}

Let $\cost(S_m)$ denote the cost of the solution to the paging problem for
$S_m$. The following two claims give an upper bound on the cost of the marking
scheme.
\begin{claim}
\label{onl:cla:markingBoundedByPaging}
	We have that
	\begin{align*}
		\cost({\mathcal{M}})
		&\le \sum_m\big(\cost(S_m)+\hmax+O(\log k)\cdot\hmax\big) \\
		&= O\left(\sum_m\cost(S_m)+\log k\cdot\hmax\right).
	\end{align*}
\end{claim}

\begin{claim}
\label{onl:cla:pagingBoundedByhmax}
	There is a randomized online algorithm for the paging problem on $S_m$ with
	(expected) cost $\cost(S_m)\le O((\log k+\log\ell)\cdot\hmax)$.
\end{claim}

Now combining the two claims above with Lemma~\ref{onl:lem:lowerILP} and the
analysis of vertex charges from Section~\ref{onl:sec:analysis}, we obtain our main
theorem.

\begin{theorem}
\label{onl:thm:randomized}
	There is a randomized algorithm with competitive ratio $O(\log\ell+\log k)$.
\end{theorem}
The constant hidden in the $O(\cdot)$-notation in the theorem is
$(1/\varepsilon)^{O(1/\varepsilon^4)}$. This follows from the same arguments
mentioned after the statement of Theorem~\ref{onl:thm:det} and the fact that
Claim~\ref{onl:cla:markingBoundedByPaging} only adds another
$(1/\varepsilon)^{O(1/\varepsilon^2)}$-factor since the number of monochromatic
configurations is $(1/\varepsilon)^{O(1/\varepsilon^2)}$.

\section{Lower Bounds}
\label{onl:sec:lbs}

In this section, we derive lower bounds on the competitive ratios for
deterministic and randomized algorithms. In particular, we show that any
deterministic algorithm must have a competitive ratio of
$\Omega(\ell\log k)$ and any randomized algorithm must have a
competitive ratio of $\Omega(\log \ell + \log k)$.

We note that the lower bounds derived in this section also apply to the model
studied by Henzinger et al.~\cite{henzinger19efficient}. Their model is
slightly more restrictive than ours in that eventually, every server must have
exactly one piece of volume 1 (resp.~$k$ in their terminology); in contrast, in
our model, servers may eventually host multiple pieces smaller than 1. However,
our lower bounds are designed such that they also fulfill the definition of the
model by Henzinger et al.

\subsection{Lower Bounds for Deterministic Algorithms}
\label{onl:sec:lbs-det}

\begin{theorem}
\label{onl:thm:lb-det}
For any $k \ge 32$ and any constant $1/k\le \varepsilon \le 1/32$ such that
$\varepsilon k$ is a power of $2$,
	any deterministic algorithm must have a competitive ratio of
	$\Omega(\ell \lg k)$.
\end{theorem}
We devote the rest of this subsection to prove the theorem. 
 
Set $m$ be a positive integer such that
$\varepsilon k = 2^m$. As $\varepsilon \le 1/32$
it follows that $k \ge 2^{m+5}$.
Fix any deterministic algorithm $\ONL$. We will show that there exists
a sequence $\sigma_{\ONL}$ of edge insertions such that the cost of the optimum
offline algorithm is $O(\varepsilon)$, while the cost of ${\ONL}$ is 
$\Omega(\varepsilon \ell \log(\varepsilon k))$. 
The sequence $\sigma_{\ONL}$ depends on ${\ONL}$, i.e., edge insertions will
depend on which servers ${\ONL}$ decides to place the pieces.

\textbf{Definitions.}
We assume that the servers are numbered sequentially.
As before, each server has a color and every vertex is colored with
the color of its initial server. For simplicity, we assume server $i$ has color $i$.
The
\emph{main server} of a color $c$ is the server that, out of all servers, currently contains the largest volume of color-$c$ vertices {\em and} whose index number of all such servers is the smallest\footnote{The difference between {\em majority} and {\em main} color is that we added the second condition to guarantee that the main server of a color is unique.)}.

A piece is called {\em single-colored} if all vertices of the piece have the same color. 
If a single-colored piece with  color $c$ is not assigned to the main server for $c$, it
is called {\em $c$-away} or simply {\em away}.
Any piece of volume at least $2\varepsilon$ is called a {\em large} piece,
all other pieces are called {\em small}.
We say {\em two pieces are merged} if there is an edge insertion connecting the two pieces.

\textbf{Initial Configuration.} Initially each of the $\ell$
servers contains one large single-colored piece of volume $2
\varepsilon$ and $(1-2\varepsilon) k$ isolated vertices, each of volume $1/k$.
The large pieces of on servers  $1$, $2$, and $3$ are called 
 \emph{special}. 
A color $c$ is {\em deficient} if the total volume of all {\em small}  $c$-away pieces is at least $\varepsilon$.

\textbf{Sequence $\sigma_{\ONL}$.}
The first two edge insertions merge the three special pieces into one (multi-colored) special piece of volume $6\varepsilon$.
As we will show {\em any} algorithm now has at least one deficient color.
Note that all small pieces are single-colored and have volume $2^0/k = 1/k$.

Now $\sigma_{\ONL}$ proceeds in {\em rounds}. We will show that there is a deficient color at the end and, thus, also at the beginning of every round. In each round
only small pieces of the same (deficient) color are merged such that their volume doubles.
As a result, all small pieces continue to be single-colored and, at the end of each round, all  small pieces of the same color have the same volume, namely
$2^i/k$ for some integer $i$, {\em except} for potentially one piece of smaller volume, which we call the {\em leftover piece}. A leftover piece is created if the number of small items of color $c$ and volume $2^i/k$ at the beginning of a round is an odd number. 
If this happens, it is merged with the leftover piece of color $c$ of the previous rounds (if it exists) to guarantee that there is always just one leftover piece of color $c$. To simplify the notation we will use the term {\em almost all small pieces of color $c$} to denote all pieces of color $c$ except the leftover piece of color $c$.

A round of $\sigma_{\ONL}$ consists of the following sequence of requests among the small
pieces: 
If there exists a deficient color $c$ such that the volume of almost all 
small pieces of color $c$ is $2^i/k$ for some integer $i$ and $2^i/k < \varepsilon,$ 
then $\sigma_{\ONL}$ contains the following steps.\\
If there are small pieces of color $c$ and volume  $2^i/k$ that are currently on different servers, they are connected by an edge, otherwise two such pieces on
the same server are connected by an edge.
Repeat this until there is at most one small piece of color $c$ of volume $2^i/k$ left.
Once this happens and if such a piece exists, it becomes a {\em leftover} piece of color $c$ and if another leftover piece of color $c$ exists from earlier rounds, the two are merged. Note that almost all pieces of color $c$ have now volume $2^{i+1}/k$ and the leftover piece has smaller volume.
If $2^{i+1}/k  \ge \varepsilon$, merge all {\em non-special} (i.e. the small and the non-special large) pieces of color $c$ and call color $c$ {\em finished}.
As long as there are at least two unfinished deficient colors, start a new round.

Once there are no more rounds we will show that there is exactly one unfinished deficient color $c^*$ left and there are at least $2^{j+3}$ small pieces  of color $c^*$ and volume $\varepsilon/2^j$ for some  integer $j \ge 1$.
Furthermore there exists the special piece of volume $6\varepsilon$ (which is not single-colored) and for every other color there exists one piece of volume 1 (if it does not belong to $\{1,2,3\}$) or of volume $1-2\varepsilon$ (if it belongs to $\{1,2,3\}$).

\textbf{Final Merging Steps.} To guarantee that each piece has volume exactly 1 at the end, the remaining pieces of volume less than 1 are now suitably merged. First $2^{j+3}$ of the pieces of color $c^*$ and volume $\varepsilon/2^j$ are merged into 3 pieces of volume $2\varepsilon$ each, the rest is merged into one piece. Then consider two cases:
If $c^* \in \{1,2,3\}$, let $c'$ and $c''$ be the other two colors of
$\{1,2,3\}$. In this case $\sigma_{\ONL}$  merges the first small piece of  volume $2\varepsilon$ of 
$c^*$ with the non-special piece of $c'$ and then merges the second small piece of volume $2\varepsilon$ of 
$c^*$ with the non-special piece of $c''$. 
Then all the remaining pieces of color $c^*$ are merged with each other and with the special piece.

If $c^* \not\in \{1,2,3\}$, then $\sigma_{\ONL}$  merges the small pieces of 
 volume $\varepsilon/2$ of color
$c^*$ with the non-special piece of color $1$ and then does the same with color
$2$ and 3. Then all the remaining (small and large)  pieces of color $c^*$ are merged with each other and with the special piece.

Note that as a consequence all piece now have volume 1.

We show first that  all the assumptions made in the description of $\sigma_{\ONL}$ hold. Specifically the next three lemmata will show that
(1) after initialization and after each round there exists a deficient color for any algorithm,  that (2) for each color $c$ almost all  small pieces
of color $c$ have volume $2^i/k$ and the leftover piece of color $c$ has volume less than $2^i/k$, and that (3) at the beginning of the final merging steps
there is exactly one unfinished deficient color left and there are at least
$2^{j+3}$ small pieces of this color that have volume $\varepsilon/2^j$ for some
integer $j \ge 1$. Then we will show that algorithm ${\ONL}$ 
has cost at least $\Omega(\varepsilon \ell \log(\varepsilon k))$ to process the sequence.

\begin{lemma}
\label{onl:lemma:split-server}
	At the beginning of each round there exists an unfinished deficient color
	for algorithm ${\ONL}$.
\end{lemma}
\begin{proof}
After initialization  and after each round there exists (1) the special piece of volume $6 \varepsilon$ that is not single-colored and (2)
for each color there exist small single-colored pieces of 
	total volume at least $1-2\varepsilon$. {\em Now suppose by contradiction that no color is deficient.} Then for each color
	$c$ the total volume of small $c$-away pieces is less than $\varepsilon$, i.e. the volume of the small pieces on the main server for $c$ is at least
	$1-3\varepsilon$. As no server can have pieces of total volume more than $1+\varepsilon$ assigned to it and $\varepsilon \le 1/8$, 
	it follows that the non-special pieces on the main server of $c$ require volume more than $(1+\varepsilon)/2$,
	and, thus, each server can be the main server for at most one color. As there are as many colors as there are server, each server is the main server for exactly one color and each color has exactly one main server. 
	
	Now consider the server $s^*$ on which the special piece 
	is placed and let it be the main server for some color $c$.
	Then the total volume of the pieces on $s^*$ is $6 \varepsilon$ for the special piece.
	If $c$ is not deficient, $s^*$ has load at least $1 - 3\varepsilon$ for the non-special pieces of color $c$.
        Thus, the server's load is at least $1+3\varepsilon$ which is not possible.
				Hence, there must exist a deficient color.	
							
Next we show that there is always an unfinished deficient color. This is trivially true after initialization as all colors are unfinished. Let us now consider the end of a round.
 Note that every color $c$ that is finished has a non-special piece of volume at least $1-2\varepsilon$ and, thus, {\em the special piece cannot be placed on the main server of a finished color $c$}.
Recall that every non-deficient color has pieces of total volume at least $1-\varepsilon$ on its main server. Thus, {\em the special piece cannot be placed the main server of any non-deficient color}. 
Thus, the special piece can  only be placed on a server that is not the main server of a finished deficient or a non-finished color. If every deficient color is finished, every color has a main server and the special piece cannot be placed on any of them. As, however, there are as many servers as there are colors, it would follow that the special piece is not placed on any server, which is not possible. Thus, there must exist a deficient unfinished color.
\end{proof}

\begin{lemma}\label{onl:lemma:smallsizes}
For each color $c$ it holds at the beginning and end of each round that almost all  small pieces of color $c$ have volume $2^i/k$ for some integer $i$ and the other small piece has even smaller volume. 
\end{lemma}
\begin{proof}
By induction on the number of rounds. The claim holds after initialization for $i = 1$ for every color. During each round for some color $c$ 
the pieces of  color $c$ and volume $2^i/k$ are merged pairwise, and the possible left-over piece of volume $2^i/k$ is merged with the leftover piece of earlier rounds, if it exists. From the induction claim it follows that the leftover piece  of earlier rounds has volume less than $2^i/k$. Thus, the resulting leftover piece has volume less than $2^{i+1}/k$.
Furthermore, the pieces of the other colors remain unchanged.
Thus, the claim follows.
\end{proof}

\begin{lemma}\label{onl:lemma:finalsteps}
At the beginning of the final merging steps
there is exactly one unfinished deficient color left and there are
at least $2^{j+3}$ small pieces have volume $\varepsilon/2^j$ for some integer $j \ge 1$.
\end{lemma}
\begin{proof}
Lemma~\ref{onl:lemma:split-server} holds after each round, thus, also after the final round. It shows that there is still at least one deficient unfinished color. As there are no more rounds, there at most one deficient unfinished color, which implies that there is exactly one deficient unfinished color.
As it is unfinished,
all its small pieces have volume less than $\varepsilon$.
For the rest of the proof we only consider small pieces of this color.

Recall that $\varepsilon k = 2^m$ and $k \ge 2^{m+5}$.
Initially there are $k - 2\varepsilon k \ge 2^{m+5}-2^{m+1}$ small pieces of volume $1/k$ each.
Let $k'$ be the largest power of 2 that is at most $k - 2\varepsilon k$.
It follows that $k -2\varepsilon k \ge k' > k/2 - \varepsilon k
\ge 2^{m+4}-2^m$. Thus, initially there are at least $k'$ small pieces of volume $1/k = \varepsilon / 2^m$ each. Let $j$ be any integer with $0 \le j \le m$ such that exactly $m-j$ rounds were executed for this color.
Thus, there are at least $k'/2^{m-j}$ pieces of volume
$2^{m-j} \varepsilon / 2^m = \varepsilon/2^j$ at the beginning of the final merging steps. 
Note that
$k'/2^{m-j} \ge (2^{m+4}-2^m)/2^{m-j} = 2^{j+4}-2^j \ge 2^{j+3}$.
Thus, there are at least $2^{j+3}$ pieces of volume $\varepsilon/2^j$.
As the color is unfinished, each small piece has volume less than $\varepsilon$, i.e. $j \ge 1$.
Thus the lemma holds.
\end{proof}

Next we analyze how many rounds are performed for a given color until it is finished.
Consider any color $c$.
The number of rounds necessary to increase the volume of almost all small pieces of color $c$  from $1/k$ to $\varepsilon$ is $\log (\varepsilon k)$ as
$\varepsilon k$ is a power of 2. Each round roughly halves the number of small pieces. Thus, we only have to show that  there are enough small pieces available initially so that $\log (\varepsilon k)$ many rounds are possible for color $c$.

\begin{lemma}
For each finished color $\log(\varepsilon k)$ many rounds are executed.
\end{lemma}
\begin{proof} Fix a color $c$ and consider in this proof only pieces of color $c$.
As  $\varepsilon \le 1/8$ and each initial small piece is a single vertex, there are $k - 2 \varepsilon k \ge 3k/4$ many such small pieces initially.
Let $k'$ be the largest power of 2 that is at most $3k/4$. Note that $k' > 3k/8$.
Thus the number of small pieces of volume
at least $\varepsilon$ is at least $k' /2^{\log (\varepsilon k)} > 3/(8\varepsilon) \ge 3$. Hence for each  finished color $\log (\varepsilon k)$ rounds will
be executed.
\end{proof}

As there are $\ell$ different colors, it suffices to show that in almost every
round algorithm ${\ONL}$ moves pieces with total volume $\Omega(\varepsilon)$ to achieve
the desired lower bound of $\Omega(\varepsilon \ell \log(\varepsilon k))$ for
the cost of ${\ONL}$.

\begin{lemma}
	In one round of the above process, except in the last round for each color,
	${\ONL}$ moves vertices with volume
	$\Omega(\varepsilon)$. In total, the algorithm moves vertices with volume 
	$\Omega(\varepsilon \ell \log(\varepsilon k))$
\end{lemma}
\begin{proof} Fix a color $c$ and only consider pieces of color $c$ in this proof.
Note that when two pieces of different servers, of volume $2^i/k$ each, are merged,
at least one of them has to change its server, resulting in a cost of $2^i/k$ for the algorithm.
We proved in Lemma~\ref{onl:lemma:split-server} that at the beginning of each round a deficient color exists.
A deficient color  has away pieces of
	total volume at least $\varepsilon$, i.e., there are small pieces of total volume
	at least
	$\varepsilon$ not on the main server. During a round, as shown by
	Lemma~\ref{onl:lemma:smallsizes}, almost all of these pieces have volume
	$2^i/k$ for some integer $i$ and their total contribution to the total volume of all away pieces of color $c$ is larger than $\varepsilon - 2^i/k$ (subtracting out the volume of  the potentially existing leftover piece of even smaller volume). Thus, as long as $\varepsilon - 2^i/k \ge \varepsilon/2$, i.e., in all but the last round, the total
	volume of all the away pieces excluding the leftover piece is larger then $\varepsilon/2$. In the following when we talk about a small piece, we mean a small piece that is {\em not} the leftover piece and we fix a round that is not the last round.
	We will show that at least $\varepsilon/2$ volume is merged by pieces on different servers in this round, resulting in at least $\varepsilon/4$ cost for the algorithm.
	
	Now consider two cases: (1) If the main server  $s^*$ contains small pieces of total volume at least $\varepsilon/2$, then every away piece can be merged with a small piece either on $s^*$ or on a different server. Thus at least $\varepsilon/2$ volume is merged by pieces on different servers. (2) If, however, the main server contains small pieces of total volume less than $\varepsilon/2$, then {\em every} server contains
small pieces of total volume less than $\varepsilon/2$. Small pieces of different servers are merged until all remaining small pieces are on the same server. However, this server has less than $\varepsilon/2$ volume of small pieces, i.e., more than
$\varepsilon/2$ volume must have been merged between different servers.
Thus in both cases the algorithm has cost at least $\varepsilon/4$.
	The second claim
	follows immediately from the discussion preceding the lemma.
\end{proof}

Next, we prove an upper bound on the volume of vertices moved by $\OPT$.
\begin{lemma}
	In total, $\OPT$ moves vertices with volume $O(\varepsilon)$.
\end{lemma}
\begin{proof}
Right at the beginning $\OPT$ places the special piece of volume $6\varepsilon$ on
server $s^*$ and moves the small pieces of color $c^*$ that are merged in
the final merging step with a different color to the 
  main server for the corresponding color. Thus,
none of the other steps cause any cost for $\OPT$. Thus,
	$\OPT$ only has cost $O(\varepsilon)$.
\end{proof}

The previous two lemmas imply a lower bound on the competitive ratio of
$\Omega(\ell \log(\varepsilon k))$ for deterministic algorithms. This finishes
the proof of the theorem.

\subsection{Lower Bounds for Randomized Algorithms}
\label{onl:sec:lbs-rand}

\begin{theorem}
	Any randomized online algorithm must have a competitive ratio of
	$\Omega(\lg \ell + \lg k)$.
\end{theorem}

\begin{proposition}
	If $\varepsilon < 1/6$, then any randomized online algorithm must have a
	competitive ratio of $\Omega(\lg \ell)$.
\end{proposition}
\begin{proof}
	We use Yao's principle~\cite{yao77probabilistic} to derive our lower bound
	and provide a randomized hard instance against a deterministic algorithm.
	The hard instance starts by merging the vertices of each server into
	monochromatic pieces of volume $2\varepsilon$ each. Now the hard instance
	arbitrarily picks three pieces with different majority colors and merges them
	into a piece of volume $6\varepsilon$ and we call this piece \emph{special}.
	Next, the hard instance proceeds in $\ell-1$ rounds. Before the first round
	all servers are \emph{unfinished}. In round $i$, the hard instance picks an
	unfinished server $s$ uniformly at random. Now the hard instance uniformly
	at random picks monochromatic pieces with color $s$ of total volume
	$1-2\varepsilon$ and merges them in arbitrary order; after that we call $s$
	\emph{finished}. When all $\ell-1$ rounds are over, a final configuration in
	which all pieces have volume $1$ is obtained as follows.  First, observe that
	there is a unique unfinished server $s^*$. Now the hard instance merges the
	special piece and monochromatic pieces of color $s^*$ of total volume
	$1-6\varepsilon$. The remaining monochromatic pieces of color $s^*$ are
	merged with the components of the finished servers from which the vertices
	of the special piece originated. All other monochromatic components of
	volume $2\varepsilon$ are merged with the large monochromatic components
	with the same color as the piece itself.

	For a given schedule, we say that an unfinished server $s$ is \emph{split}
	if monochromatic pieces with color $s$ and of volume at least $\varepsilon$
	are not scheduled on $s$.  Now observe that after each round there exists a
	server which is split: First, observe that none of the finished servers can
	store its monochromatic piece of volume $1-2\varepsilon$ together with the
	special piece of volume $6\varepsilon$. Now if none of the (unfinished)
	servers was split, one of them would contain all of its monochromatic pieces
	of total volume at least $1-2\varepsilon$ together with the special piece of
	volume $6\varepsilon$.  Thus, the total load of the server is
	$1+4\varepsilon$ which is not a valid schedule.

	Next, we show that if a server $s$ is split, then the algorithm has moved
	monochromatic pieces with color $s$ of volume at least $\varepsilon$: First,
	suppose the algorithm has scheduled all monochromatic pieces of color $s$ on
	some server $s' \neq s$.  Then the algorithm has paid at least
	$1-2\varepsilon\geq\varepsilon$ to move the monochromatic pieces of color
	$s$ to $s'$.  Second, suppose the monochromatic pieces of color $s$ are
	scheduled on at least two different servers.  Then the algorithm must have
	moved at least one monochromatic piece of color $s$ away from $s$.  Since
	$s$ is unfinished and all monochromatic pieces of $s$ have volume
	$2\varepsilon$, the algorithm has paid at least $\varepsilon$ for moving
	monochromatic pieces of color~$s$. 

	Now we analyze the cost paid by the algorithm. Observe that before round $i$
	there are $\ell-i+1$ unfinished servers and at least one of them is split.
	Let $s$ be a split server.  Thus with probability $1/(\ell-i+1)$ the hard
	instance picks the split server $s$.  It follows from the previous claims
	that the algorithm paid at least $\varepsilon$ to move pieces of color~$s$.
	Since the above arguments hold for each round, the total expected cost of
	the algorithm is
	\begin{align*}
		\sum_{i=1}^{\ell-1} \varepsilon \frac{1}{\ell-i+1}
		= \sum_{i=2}^\ell \varepsilon \frac{1}{i}
		= \Omega(\varepsilon \lg \ell).
	\end{align*}

	Next, observe that $\OPT$ never moves more than $O(\varepsilon)$ volume:
	Indeed, the hard instance only merges pieces in which all vertices have the
	same color except when (1)~creating the special piece of volume
	$O(\varepsilon)$, (2)~merging the special piece with the vertices from $s^*$
	and (3)~merging the small pieces from $s^*$ with the large pieces of the
	servers from which the special piece originated. All of these steps can be
	performed by only moving volume $O(\varepsilon)$.
	
	Thus, the competitive ratio is $\Omega(\lg \ell)$.
\end{proof}

\begin{proposition}
\label{onl:prop:lb-lg-k}
	Any randomized algorithm must have a competitive ratio of at least
	$\Omega(\lg k)$.
\end{proposition}
\begin{proof}
	We use Yao's principle~\cite{yao77probabilistic} to derive our lower bound
	and provide a random instance against a deterministic algorithm. In the
	instance all pieces initially have volume $1/k$, i.e., the pieces consist of
	single vertices. The lower bounds proceeds in $\lg k $ rounds.
	In each round, we pick a perfect matching between all pieces uniformly at
	random.  Thus, after $i$ rounds, all pieces have volume $2^i/k$. Note that
	after $\lg k$ rounds all pieces have volume $1$ and we have obtained a valid
	final configuration.

	We claim that in each round the algorithm has to move volume
	$\Omega(\ell)$. Suppose we are currently in round $i$. Now consider two
	pieces $\ps$ and $\pl$ which are merged during a single round. Then the
	probability that $\ps$ and $\pl$ are assigned to different servers is
	$\Omega( (\ell-1)/\ell ) = \Omega(1)$.  Furthermore, observe that each piece
	has volume $2^i/k$ and in total there are $n/2^i$ pieces.  Now by linearity
	of expectation we obtain that the expected volume moved by the algorithm in
	round $i$ is $\Omega( 2^i/k \cdot n/2^i ) = \Omega(\ell)$.

	Next, observe that the total cost paid by the algorithm is
	$\Omega(\ell\cdot\lg k)$ since there are $\log k$ rounds. Furthermore,
	$\OPT$ never moves volume more than $O(\ell)$ because it moves each vertex
	at most once. Thus, the competitive ratio is $\Omega(\log k)$.
\end{proof}

\section*{Acknowledgements}
Monika Henzinger is supported by the Austrian Science Fund (FWF): 33775-N.
Stefan Neumann gratefully acknowledges the financial support from the Doctoral Programme “Vienna Graduate School on Computational Optimization” which is funded by the Austrian Science Fund (FWF, project no. W1260-N35) and from the European Research Council under the European Community’s Seventh Framework Programme (FP7/2007-2013) / ERC grant agreement No.~340506.
Stefan Schmid is supported by the European Research Council under the European Community's  Horizon 2020 Programme / ERC grant agreement No.~864228 and by the Austrian Science Fund (FWF): 33775-N.

{
\bibliographystyle{plain} 
\bibliography{main}
}

\appendix
\section{Omitted Proofs}
\label{onl:sec:omitted}

\addtocontents{toc}{\setcounter{tocdepth}{-10}}

\subsection{Proof of Claim~\ref{onl:claim:numberOfConfigurations}}
	Both $r$ and $m$ are vectors with one entry for each size class.
	As argued in Section~\ref{onl:sec:preliminaries}, there are only $O(1)$
	different size classes. By definition of $r$ and $m$, each entry
	is a multiple of $\delta = O(1/\varepsilon^2)$ between $0$ and $1+\gamma$.
	Therefore, there are only $O(1)$ choices for each entry of $r$ and $m$.

\subsection{Proof of Lemma~\ref{onl:lem:lowerILP}}
	Consider the solution of $\OPT$, i.e., an assignment of the pieces to
	servers at minimum moving cost. We show that this assignment implies an ILP
	solution with small cost.  In particular, for each server $s$ for which $\OPT$
	assigns at least a $(1-\gamma+\delta)$-fraction of its source vertices on
	$s$ itself, we show that we can construct an ordinary configuration.  For
	all other servers (for which $\OPT$ assigns at least a
	$(\gamma-\delta)$-fraction of their source vertices on other servers), we
	assign an extraordinary configuration.  Note for each server which has
	obtained an extraordinary configuration in the previous construction, $\OPT$
	has cost at least $\gamma-\delta$.  Since the ILP minimizes the total number
	of extraordinary configurations, we get that $\cost(\OPT)\ge (\gamma-\delta)h$.

	We now construct the configurations for ordinary servers.  Fix a server $s$
	for which $\OPT$ assigns at least a $(1-\gamma+\delta)$-fraction of its
	source vertices on $s$ itself. Let $v$ be a vector such that $v_i$ is the
	real volume of pieces in the $i$-th class, that $\OPT$ assigned to $s$. Note
	that the $v_i$ are \emph{not} rounded to multiples of $\delta$ and that
	$\|v\|=1$.  We construct a $\gamma$-valid reservation vector~$r$ for
	server~$s$ as follows.

	For every piece~$p$ in class~$i$ that is assigned on~$s$, we decrease $v_i$
	by the \emph{uncommitted} volume of the piece and add this value to $v_0$.
	For a piece~$p$ of class~$i$ which is not assigned on~$s$, we just consider
	the source vertices of~$s$ in~$p$. We increase~$v_i$ by the committed volume
	of these vertices and $v_0$ by the uncommitted volume.

	The above modifications did not increase $\|v\|_1$ if the piece $p$ is
	assigned on $s$. The total increase of $\|v\|_1$ due to pieces not
	assigned on $s$ can be at most $\gamma-\delta$ because each such increase is
	caused by an element with source-server $s$ that is not assigned on~$s$ (and
	the total volume of such vertices is only $\gamma-\delta$).
	Now we round $v_0$ up to the nearest multiple of $\delta$. This
	increases $\|v\|_1$ by at most $\delta$. Therefore, $\|v\|_1\le1+\gamma$.

	We use the final vector $v$ as our reservation
	vector $r$.
	By construction $r_0$ is at least as large as the total uncommitted
	volume
	with source server $s$, i.e., $r_0\ge m_{s0}$. Further, $r_i$ is at least as
	large as the committed volume of monochromatic pieces with source $s$, i.e.,
	$r_i\ge m_{si}$. Finally, the $r_i$ value is still as large as the total
	committed volume of class $i$ elements assigned on $s$. This means that the
	first set of constraints in the ILP still holds.

\subsection{Proof of Lemma~\ref{onl:lem:augmentation}}
	The committed volume in large pieces can be at most $1+\gamma-r_0$ because
	$\|r\|_1\le 1+\gamma$. The uncommitted volume scheduled on $s$ can be at most
	$r_0+14\epsilon$ due to Property~\ref{onl:p:two}.

\subsection{Proof of Lemma~\ref{onl:lem:slack}}
	We want to show that $\slack(s) = r_{s0}-v_u(s)\geq -14\epsilon$. 
	We prove the lemma by contradiction, i.e., assume that $\slack(s)<-14\epsilon$. 
	We will use the following claim, which we prove at the end of the subsection.
	\begin{claim}
	\label{onl:cla:budget}
		$\mathit{budget(s)}\ge v_u(s)-r_{s0}-2\epsilon = -\slack(s)-2\varepsilon$.
	\end{claim}

	First, we argue that the total volume of small pieces is $r_{s0}+4\varepsilon$.
	Let $\bar{m}_{s0}$ denote the total uncommitted volume for $s$
	\emph{not} rounded up, i.e., $m_{s0}=\lceil\bar{m}_{s0}\rceil_\delta$.
	Since $\slack(s)<-14\varepsilon$,
	the above claim gives that the budget is at least $12\epsilon\ge 2\epsilon$, i.e., it is
	larger than the volume of small pieces.
	Thus, the only reason to not perform an eviction is that all small pieces scheduled
	on $s$ are close to being monochromatic for $s$ and $s$ is ordinary. But the
	total volume of such pieces can be at most
	$$\bar{m}_{s0}/(1-2\epsilon)\le (1+4\epsilon)\bar{m}_{s0}\le \bar{m}_{s0}+4\epsilon\le r_{s0}+4\epsilon,$$
	where we used $\epsilon\le1/4$ and that $s$ has an ordinary configuration.

	Second, the following claim gives a bound of $10\epsilon$ on the
	total uncommitted volume in large pieces; we prove the claim at the end of
	this subsection.
	\begin{claim}
	\label{onl:cla:uncommitted}
		If $\gamma \leq 1$, then the total uncommitted volume in large
		pieces at a server $s$ is at most $10\epsilon$.
	\end{claim}

	We conclude that Property~\ref{onl:p:two} holds, i.e., $v_u(s)\leq r_{s0}+14\epsilon$.
	However, this is a contradiction to our assumption that
	$r_{s0}-v_u(s)=\slack(s)<-14\varepsilon$ since this inequality implies
	$v_u(s)>r_{s0}+14\varepsilon$.

	\begin{proof}[Proof of Claim~\ref{onl:cla:budget}]
		In the initial state of the algorithm $\mathit{budget}(s)=0$, $v_u(s)=1$,
		$m_{s0}=r_{s0}=\lceil 1\rceil_{\delta}$. Thus, the statement holds.
		The following operations affect the slack:
		\begin{itemize}[noitemsep,label=--]
			\item A piece $p$ is moved to $s$ outside of the balancing procedure. This
			increases both sides of the equation by $|p|_u$.
			\item A small piece $p$ is moved to $s$ inside of the balancing procedure.
			This is only performed if the slack on $s$ is non-negative. Note that the
			move causes $v_u(s)$ to increase by at most $|p|\leq 2\varepsilon$
			and thus it holds that $-\slack(s) \leq 2\varepsilon$ after the move is
			completed.
			Hence the equation holds since $\mathit{budget(s)}\geq0$
			at all times.
			\item A piece is moved away from $s$ outside of the balancing procedure. This
			only decreases $v_u(s)$.
			\item A small piece is moved away from $s$ inside the balancing procedure. This
			decreases both sides of the equation by the volume of the piece.
			\item When adjusting the schedule by the generic routine (see
			Section~\ref{onl:sec:variantGeneric}) and a new configuration with a smaller
			$r_{s0}$-value is assigned to $s$, then the eviction budget is increased by the
			change in $r_{s0}$. Hence, both sides of the equation increase by the same
			amount.
		\end{itemize}
		As no operation can make the equation invalid it holds throughout the algorithm.

		Note that when we use Special Variant~A for adjusting the schedule (see
		Section~\ref{onl:sec:variantA}) $r_{s0}$ will be decreased but the committed volume
		scheduled on $s$ will be decreased by the same amount. This means the slack
		does not change in this case.
	\end{proof}

	\begin{proof}[Proof of Claim~\ref{onl:cla:uncommitted}]
		Consider the set $L$ of large pieces that exclude the piece $\pm$ that
		resulted from the last merge-operation.
		Let $\bar{v}_c$ and $\bar{v}_u$ denote the total committed and
		uncommitted volume for pieces in $L$ that are scheduled on $s$.
		Pieces in $L$  have volume at least~$\epsilon$ and uncommitted volume at
		most $2\delta$ (Invariants~\ref{onl:inv:largepiecelarge}
		and~\ref{onl:inv:uncommitted-volume}). Therefore, the factor
		$f:=\bar{v}_u/(\bar{v}_c+\bar{v}_u)$ between uncommitted and real volume
		of
		these pieces is at most $2\delta/\epsilon$. We can derive a bound on the total
		uncommitted volume for pieces in~$L$ as follows:
		\begin{equation*}
		\bar{v}_c+\bar{v}_u
		= \big(1+\tfrac{f}{1-f}\big)\bar{v}_c
		\le (1+2f)\bar{v}_c\enspace,
		\end{equation*}
		where we use $f\le 2\delta/\epsilon\le 1/2$, which holds because
		$\delta\le\epsilon^2$ and $\epsilon\le1/4$. To obtain a bound
		on the total uncommitted volume in large pieces, we have to also consider
		the piece $\pm$. For this piece we have $|\pm|_u\le \epsilon+2\delta$ according
		to Invariant~\ref{onl:inv:uncommitted-volume}. We get
		\begin{equation*}
		\begin{split}
		v_u=\bar{v}_u+|\pm|_u
		\le 2f\bar{v}_c+(\epsilon+2\delta)
		\le\tfrac{4\delta}{\epsilon}(1+\gamma)+2\epsilon
		\le10\epsilon\enspace,
		\end{split}
		\end{equation*}
		where the second step uses that $|\pm|_u\le \epsilon+2\delta$ due to
		Invariant~\ref{onl:inv:uncommitted-volume}, the third step uses that the committed volume at $s$
		is at most $\|r\|_1\le1+\gamma$, and the final step uses $\gamma\le 1$.
	\end{proof}

\subsection{Proof of Lemma~\ref{onl:lem:adjustLemma}}
	The third claim holds since only the servers for which one of the pieces
	$\ps,\pl,$ or $\pm$ is monochromatic can change their source vector.

	Next, consider the first claim. The cost for the second step in the scheme
	is at most $O(|\mathcal{A}\cup \mathcal{C}|)$ because only pieces that are
	monochromatic for a server in $\mathcal{A}\cup \mathcal{C}$ move and the
	total volume of such pieces is $O(|\mathcal{A}\cup\mathcal{C}|)$.
	The cost for the third step is also at most $O(|\mathcal{A}\cup \mathcal{C}|)$ because
	only pieces that were scheduled on servers from $\mathcal{A}\cup \mathcal{C}$ in $S$
	move and each server was assigned pieces of total volume at most
	$1+O(\varepsilon)=O(1)$ (note that the number of pieces moved does not
	matter here, but only the total volume of these pieces which is
	$O(|\mathcal{A}\cup\mathcal{C}|)$).
	This gives the first claim since
	$$|\mathcal{A}\cup \mathcal{C}|\le 3+|\mathcal{B}|+\|x-x'\|_1=O(1+D+\|x-x'\|_1).$$

	The rest of the proof is devoted to the proof of the second claim.
	We use the following general result
	about the sensitivity of optimal ILP solutions. It states that a small change
	in the constraint vector of the ILP implies only a small change in the optimal solution $x$.
	\begin{theorem}[{\cite[Corollary~17.2a]{schrijver99theory}}]
	\label{onl:thm:sensitivity}
		Let $A$ be an integral $n_r\times n_c$ matrix, such that each subdeterminant of $A$
		is at most $\Delta$ in absolute value; let $b'$ and $b''$ be column
		$n_r$-vectors, and let $c$ be a row $n_c$-vector. Suppose
		$\max\{cx\mid Ax\le b';  \text{$x$ integral}\}$ and
		$\max\{cx\mid Ax\le b''; \text{$x$ integral}\}$ are finite. Then for each
		optimum solution $z'$ of the first maximum there exists an optimum solution
		$z''$ of the second maximum such that $\|z'-z''\|_\infty \le  n_c\Delta(\|b'-b''\|_\infty+2)$.
	\end{theorem}

	To apply the theorem, we bound how much the constraint vector in our ILP
	changes. Every change in the value of 
	$Z_m$ represents a change in the source configuration of some server. Hence,
	$\|Z-Z'\|\le D$. Every change in
	a value $V_i$ represents a change in the committed volume of a piece. For a merge
	operation there are at most 3 pieces that change their committed volume
	(pieces $\ps,\pl,pm$ in a merge-operation). For a commit-operation only the
	piece $\pm$ executing the commit changes its committed volume. Hence,
	$\|V-V'\|_1\le 3/\delta$. Overall the RHS vector in the ILP changes by
	$O(1+D)$.

	Let $A$ denote the matrix that defines the ILP. Then the number of columns 
	$n_c$ is the number of configurations $(r,m)$. This is constant due to
	Claim~\ref{onl:claim:numberOfConfigurations}.

	An entry in the matrix $A$ is either $0$, $1$ or 
	$r_i/\delta\le\|r\|_1/\delta\le (1+\gamma)/\delta$. Hence,
	$a_{\max}:=(1+\gamma)/\delta$ is an upper bound for the absolute value
	of entries in $A$. As the number of columns is $n_c$, we can use
	Hadamard's inequality to get a bound of $\Delta\le
	n_c^{n_c/2}a_{\max}^{n_c}=O(1)$ on
	the value of any subdeterminant.

	Now, Theorem~\ref{onl:thm:sensitivity} gives that we can find an optimum
	ILP solution $x'$ for $\mathcal{P}'$ with $\|x-x'\|_\infty\le O(1+\|b-b'\|_\infty)$.
	As the vectors $x$ and $x'$ have a constant number of entries, we also get
	$\|x-x'\|_1\le O(1+\|b-b'\|_1)=O(1+D)$, as desired.

\subsection{Details of Special Variant B for Merging Monochromatic Pieces}
\label{app:onl:sec:variantB}
Suppose we perform a commit-operation for a monochromatic piece $\pm$ that is
located at an ordinary server~$s$. Then $\OPT$ may not experience any cost.
Therefore, we present a special variant for adjusting the schedule that also
induces no cost.  We perform a routine similar to Special Variant~A.
Recall that in a commit for a monochromatic piece, we commit volume exactly
$\delta$ and all of the committed volume has color $s$.

Let $i$ and $i'$ denote the old and the new class of
the piece, respectively. Then the source vector vector $m_s$ of $s$ changes as
follows:
\begin{equation*}
\begin{array}{lclll}
m_{si}'  &:= &m_{si}  &- &|\pm|_c   \\
m_{si'}' &:= &m_{si'} &+ &|\pm|_c+\delta\\
m_{s0}'  &:= &m_{s0}  &- &\delta   \\
\end{array}.
\end{equation*}
Note that the above is also correct for the case that $i=0$ ($\pm$ small)
because then $|\pm|_c=0$.
The new ILP is obtained by setting
\begin{equation*}
\begin{array}{lclll}
Z_{m_s}'        &:= &Z_{m_s} &- &1\\
Z_{m_s'}'       &:= &Z_{m_s'}  &+ &1\\
\end{array}
\text{~~~and~~~}
\begin{array}{lclll}
V_{si}'      &:= &V_{si}  &- &|\pm|_c   \\
V_{si'}'     &:= &V_{si'} &+ &|\pm|_c+\delta\\
V_{s0}'      &:= &V_{s0}  &- &\delta\\
\end{array}.
\end{equation*}
We obtain a solution to the new ILP by adjusting the reservation vector $r$ at the
server $s$ where $\pm$ is scheduled (recall that $s$ is ordinary, i.e., the
monochromatic piece $\pm$ must be located at $s$):
\begin{equation*}
\begin{array}{lclll}
r_{si}'  &:= &r_{si}  &- &|\pm|_c   \\
r_{si'}'  &:= &r_{si'} &+ &|\pm|_c+\delta\\
r_{s0}'  &:= &r_{s0}  &- &\delta   \\
\end{array}.
\end{equation*}
Observe that we reduce $r_{s0}$ by $\delta$. Usually, if we reduce $r_{s0}$ we
increase the eviction budget of $s$, so that $s$ can evict small pieces in case
the uncommitted volume scheduled on $s$ is larger than $r_{s0}+14\epsilon$.
However, here this is not necessary because the commit also decreases the
uncommitted volume that is scheduled on $s$ by $\delta$.

Observe that $\|r_s\|\le 1+\gamma$ implies $\|r_s'\|\le 1+\gamma$, i.e.,
$r'$ is $\gamma$-valid.
The new solution $x'$ is 
\begin{equation*}
\begin{array}{lclll}
  x_{(r_s,m_s)}'  &:= &x_{(r_s,m_s)}   &- &1\\
  x_{(r_s',m_s')}' &:= &x_{(r_s',m_s')}  &+ &1\\
\end{array}.
\end{equation*}
Note that the configuration $(r_s',m_s')$ is ordinary because $(r_s,m_s)$ is ordinary.
Overall only a single server changed its configuration and this server keeps
an ordinary configuration. Therefore the number of extraordinary
configurations did not increase and we still have an optimum solution to the
ILP. See~Lemma~\ref{onl:lem:ilp-solution-still-optimal-commit} in
Section~\ref{onl:sec:ilp-solution-still-optimal}
for a formal proof that the new solution is indeed optimal.

Also the old schedule still respects this new ILP solution. Therefore we do
not experience any cost.

\subsection{Proof of Claim~\ref{onl:cla:chargeIandII}}
	A vertex $v$ only experiences a \TypeII charge if the piece that it is
	contained in just increased its volume by at least a factor of $2$. This can
	happen at most $\log k$ times and therefore the total \TypeII charge at a
	vertex is at most $O(\log k\cdot\size(v)/\delta)$.

	Fix a vertex $v$ and define $a_i := |\ps|$ and $b_i:=|\pm|$ at the time of
	the $i$-th \TypeI charge to vertex $v$. Then the total \TypeI charge to $v$
	is
	\begin{equation*}
	\frac{2}{\delta}\size(v)\sum_{i\ge1}\frac{a_i}{b_i}\enspace.
	\end{equation*}
	To estimate $\sum_i\frac{a_i}{b_i}$, we use the
	fact that $b_i\ge b_{i-1}+a_i$ and that each $a_i$
	is a multiple of $1/k$. We define $A_i:=\sum_{j=1}^ia_j$.
	This gives
	\begin{equation*}
	\sum_{i\ge1}\frac{a_i}{b_i}
	\le\sum_{i\ge1}\frac{a_i}{A_i}
	=\sum_{i\ge 1}\sum_{j=1}^{ka_i}\frac{1}{kA_i}
	\le\sum_{i\ge 1}\sum_{j=1}^{ka_i}\frac{1}{kA_i-j}
	=\sum_{j=1}^{kA_t-1}\frac{1}{j}\enspace,
	\end{equation*}
	where $t$ denotes the total number of charges to vertex $v$.
	Since, $A_t$ is at most $1$ we get that the sum is $O(\log k)$. This gives
	a total \TypeII charge of at most $O(\log k\cdot\size(v)/\delta)$.

\subsection{Proof of Lemma~\ref{onl:lem:typeIII}}
	Let $p$ be monochromatic for $s$ and $p'$ be monochromatic for $s'$, i.e.,
	the \TypeIII charges occur because we have to move $p$ to $s$ and $p'$ to $s'$.
	We distinguish two cases.

	First assume that $s\neq s'$ and let $v(s)$ denote the volume of vertices
	of color $s$ in~$p'$. Then
	\begin{equation*}
	(1-\epsilon)|p| \le v(s) \le \epsilon|p'|\enspace,
	\end{equation*}
	where the first inequality follows because $p$ is monochromatic for $s$ and the
	second because $p'$ is monochromatic for $s'$. We get
	that $|p'|\ge (1-\epsilon)/\epsilon\cdot|p|\ge (1+\epsilon)|p|$, where the last
	inequality holds for $\epsilon\le \sqrt{ 2}-1$, which holds as $\epsilon\le1/4$.

	Now, assume that $s=s'$. This means that $p$ was moved to server $s$,
	subsequently a piece $p''\supset p$ was moved away from $s$,
	and in the end $p'\supset p''$ was moved back to $s$.

	The server $s$ cannot be extraordinary at time $t$ as then there would be no need to move
	$p$ to $s$ and no \TypeIII charge would occur. Also, $s$ cannot become
	extraordinary between time $t$ and $t'$ because then the \TypeIII charge at $t$
	would be canceled. Hence, $s$ is ordinary.

	The only reason for moving $p''$ away from $s$ is one of the following:
	\begin{itemize}
	\item The eviction routine does it. As $s$ is ordinary, this 
	routine only moves $p''$ if at
	most a $(1-2\epsilon)$-fraction of its vertices have color $s$. Let $\bar{v}(s)$
	denote the volume of vertices in $p''$ that have a color different from $s$. 
	Then
	\begin{equation*}
	2\epsilon|p''|\le \bar{v}(s)\le \epsilon|p| + |p''|-|p|\enspace,
	\end{equation*}
	because at most a volume of $\epsilon|p|$ did not have color $s$ at time $t$ and
	after that at most a volume of $|p''|-|p|$ has been added. This gives
	$|p'|\geq |p''| \ge (1-\epsilon)/(1-2\epsilon)\cdot|p|\ge (1+\epsilon)|p|$.

	\item $p''$ is moved just before being merged with a (larger) piece $\pl$. 
	Let $\pm$ denote the piece obtained by merging $p''$ with $\pl$. Then
	$|p| \le |p''| \le |\pm|/2 \le |p'|/2$, which gives $|p'|\ge 2|p|$.

	\item $p''$ is a large piece and we determine a new location for it when
	resolving the ILP and adjusting the schedule. However, then we only move it
	away if it is not monochromatic for $s$. But this is a contradiction to the
	fact that $p'$
	is monochromatic for $s$ at time $t'$ because large pieces cannot become
	monochromatic by merging them with other pieces.
	\end{itemize}

\subsection{Proof of Corollary~\ref{onl:cor:chargeIII}}
	Since, between any two uncanceled \TypeIII charges to a vertex $v$ the volume of the
	piece that $v$ is contained in must grow by a factor of $1+\epsilon$, there can be
	at most $O(\log k)$ such charges, each charging $4/\delta$.

\subsection{Proof of Claim~\ref{onl:cla:chargeIV}}
	Whenever we perform a \TypeIV charge for a vertex $v$, the piece that $v$ is
	contained in just increased its volume by $|\ps|\ge\delta$. This can happen at
	most $1/\delta$ times.

\subsection{Proof of Lemma~{onl:lem:monocommit}}
	We analyze $\Nnm(s)$ for a fixed server $s$. Observe that a
	non-monochromatic commit can
	only change the entry $m_{s0}$ in a source vector as the other entries concern
	volume of monochromatic pieces, which does not change due to the commit. Let
	$v_s$ denote the uncommitted volume of server $s$, i.e.,
	$m_{s0}=\lceil v_s\rceil_\delta$ where $\lceil\cdot\rceil_\delta$ denotes
	the operation of rounding up to a multiple of $\delta$. Let $\xi_i(s)$ denote
	the reduction in $v_s$ caused by the $i$-th commit-operation (monochromatic or
	non-monochromatic).

	Then the total reduction in $v_s$ throughout the algorithm is exactly $\sum_i\xi_i(s)$.
	Observe that in the beginning of the algorithm $v_s=1$. Furthermore, by
	choice of $\delta$ we have that $\lceil 1\rceil_\delta-1\leq \delta/2$
	which is equivalent to $1-\lceil 1-\delta\rceil_\delta \geq \delta/2$. This ensures that
	the first change in $m_{s0}$ can occur only after $v_s$ decreased by at least
	$\delta/2$. Every other change occurs after $v_s$ decreased by an additional
	value of $\delta$. Hence, if at least one change in $m_{s0}$ occurs, we have
	\begin{equation*}
	\delta(\Nm(s)+\Nnm(s)-1)+\delta/2 
	\le
	\sum_{i}\xi_i(s)=
	\sum_{i\in \Im}\xi_i(s)
	+
	\sum_{i\in \Inm}\xi_i(s)\enspace,
	\end{equation*}
	where
	$\Im$
	and $\Inm$ denote the index set of monochromatic and non-monochromatic commits,
	respectively, and $\Nm(s)$ denotes the number of times that a monochromatic commit causes
	a change in the source vector of $s$.
	For a monochromatic commit $\xi_i(s)$ is either $0$ or $\delta$. This gives that
	\begin{equation}
	\label{onl:eqn:help}
	\begin{split}
	\delta\Nnm(s)-\delta/2
	\le
	\sum_{i\in \Inm}\xi_i(s)\enspace.
	\end{split}
	\end{equation}
	For $\Nnm(s)\ge 1$ we obtain
	\begin{equation*}
	\begin{split}
	\Nnm(s)
	\le
	2\Nnm(s)-1
	\le
	2\sum_{i\in \Inm}\xi_i(s)/\delta\enspace,
	\end{split}
	\end{equation*}
	by multiplying Equation~\ref{onl:eqn:help} with $2/\delta$.
	Next, observe that 
	$\sum_s\xi_i(s)=\delta$ since
	every non-monochromatic commit switches
	a volume of exactly $\delta$ from uncommitted to committed. Hence,
	summing over all servers gives
	\begin{equation*}
	\begin{split}
	\sum_{s:\Nnm(s)\ge 1}\Nnm(s)\le
	2\sum_s\sum_{i\in \Inm}\xi_i(s)/\delta
	=
	2\sum_{i\in \Inm}\sum_s\xi_i(s)/\delta
	=2N~~.
	\end{split}
	\end{equation*}

\subsection{Proof of Claim~\ref{onl:cla:chargeV}}
	A vertex $v$ can participate in at most $1/\delta$ commit-operations as each such
	operation increases the committed volume of the piece that $v$ is contained in
	by $\delta$ and the committed volume of a piece can be at most $1$. For each
	monochromatic commit-operation the vertex is charged $3\CtypeV/\delta\cdot\size(v)$. This
	gives the lemma.

\subsection{Proof of Lemma~\ref{onl:lem:totalExtraCharge}}
Before we prove Lemma~\ref{onl:lem:totalExtraCharge}, we prove the following claim
which bounds the total extra charge for a single server.
\begin{claim}
\label{onl:cla:extraChargeServer}
	The total extra charge for a server $s$ is at most $O(\log k)$.
\end{claim}
\begin{proof}
	Suppose whenever we perform an extra charge of Type~\ref{onl:extra:A} for server
	$s$, we place this charge on the vertices in the smaller piece $\ps$ such
	that each
	such vertex $v$ receives a charge of $O(1)\cdot\size(v)$. Then the charge that
	can accumulate at a vertex can be at most $O(\log k)\cdot\size(v)$, since a
	vertex can be on the smaller side of a merge-operation at most $O(\log k)$
	times. Since the volume of all vertices originating on $s$ is $1$ we obtain
	$O(\log k)$ for the extra charge of Type~\ref{onl:extra:A}.

	Clearly the extra charge of Type~\ref{onl:extra:B} can be at most $O(1)$ for any
	server.
\end{proof}
\begin{proof}[Proof of Lemma~\ref{onl:lem:totalExtraCharge}]
	For $\hmax=0$, there are no extraordinary servers and thus we do not make any
	extra charges. For $\hmax\geq1$, the lemma follows from
	Claim~\ref{onl:cla:extraChargeServer} since there are only $\ell$ servers.
\end{proof}

\subsection{Proof of Lemma~\ref{onl:lem:vertexcharge}}
	We have the following vertex charges:
	\begin{itemize}
		\item \TypeI charge and \TypeII charge:\\
		Accumulates to at most $O(\log(k)/\delta)\cdot\size(v)$ according to Claim~\ref{onl:cla:chargeIandII}.
		\item \TypeIII charge:\\
		Accumulates to at most $O(\log k)\cdot\size(v)$ according to Corollary~\ref{onl:cor:chargeIII}.
		\item \TypeIV charge:\\
		Accumulates to at most $\CtypeIV/\delta\cdot\size(v)$ according to Claim~\ref{onl:cla:chargeIV}.
		\item \TypeV charge:\\
		Accumulates to at most $3\CtypeV/\delta^2\cdot\size(v)$ according to Claim~\ref{onl:cla:chargeV}.
	\end{itemize}
	This gives the lemma.

\subsection{Proof of Theorem~\ref{onl:thm:det}}
	All the cost is either charged by successful vertex charges or by extra
	charges. Therefore, the cost of the online algorithm is at most the total
	successful vertex charge plus the extra charge.

	Lemma~\ref{onl:lem:totalExtraCharge} together with the observation that no extra
	charge is performed if $\hmax=0$ gives that the total extra charge is at most
	$O(\ell\log k\cdot\hmax)=O(\ell\log k)\cdot\cost(\OPT)$.

	Lemma~\ref{onl:lem:vertexcharge} shows that the maximum vertex charge
	$\chargemax$
	(successful or unsuccessful) made to a vertex $v$ is
	at most $O(\log k)\cdot\size(v)=O(\log k)\cdot1/k$. From Lemma~\ref{onl:lem:lowerVertex} we obtain
	\begin{equation*}
	\chargesucc\le
	(k\cdot\chargemax)\cdot\cost(\OPT)=O(\log k)\cdot\cost(\OPT).
	\end{equation*}

	Hence, we can bound the total extra charge and the total successful vertex
	charge by $O(\ell\log k)\cdot\cost(\OPT)$. This gives a competitive ratio of
	$O(\ell\log k)$.

\subsection{Proof of Theorem~\ref{onl:thm:det-large-epsilon}}
	The theorem follows immediately from the paragraphs preceding the theorem
	statement.

\subsection{Proof of Observation~\ref{onl:obs:monotone}}
	$m_s$ changes because of the following operations:
	\begin{itemize}
	\item
	A monochromatic commit on a piece $p$ in class $i$ executes $m_{s0}':=m_{s0}-\delta$,
	$m_{si}':=m_{si}-|p|_u$, and $m_{s,i+1}':=m_{s,i+1}+|p|_u+\delta$.
	\item
	A non-monochromatic commit may decrease $m_{s0}$ but does not increase any entry.
	\item
	A monochromatic merge
	executes  $m_{si_s}':=m_{si_s}-|\ps|_u$, $m_{si_\ell}':=m_{si_\ell}-|\pl|_u$,
	and $m_{si_m}':=m_{si_m}+|\pm|_u$.
	\item If the piece $\pm$ is not monochromatic after a merge (and $\pm$ is
	large) then entries in $m_s$ are only reduced.
	\item A merge between small pieces does not change $m_s$ (only the following
	commit may do so).
	\end{itemize}
	This means operations either reduce entries in $m_s$ or they shift the mass of
	$m_s$ to higher entries while not changing $\|m_s\|$. This gives the observation.

\subsection{\texorpdfstring{Marking Scheme Without Knowledge of $\hmax$}{Marking Scheme Without Knowledge of Hmax}}
\label{onl:sec:marking-doubling}
	In Section~\ref{onl:sec:marking-scheme} we showed how we can construct a marking
	scheme when $\hmax$ is known in advance. We now argue how this assumption can be
	dropped using a simple doubling trick.

	In particular, when the algorithm starts and no edges were revealed, we set
	$\hmax=0$. After that, when the object function of ILP is at least $1$ for the
	first time, we set $\hmax=1$ and run the marking scheme with fixed $\hmax$ from
	Section~\ref{onl:sec:marking-scheme}. After that, whenever the objective function of the
	ILP is larger than $\hmax$, we double the value of $\hmax$ and restart the
	marking scheme from Section~\ref{onl:sec:marking-scheme}.

	It remains to analyze the cost of this scheme. First, let $\hmax^{\text{final}}$
	denote the final value of $\hmax$ used by the above procedure when the algorithm
	stops and let $\hmax^*$ denote the highest objective function value of the ILP
	at any point in time. Now observe that $\hmax^{\text{final}} \leq 2 \hmax^*$.
	Second, note that for fixed $\hmax$, Claim~\ref{onl:cla:markingBoundedByPaging}
	and Claim~\ref{onl:cla:pagingBoundedByhmax} imply that
	$\cost({\mathcal{M}}) \leq O( (\lg k + \lg \ell) \hmax)$.
	Third, observe that the first two points imply that the total cost paid by the
	above procedure is
	\begin{align*}
		\sum_{i=0}^{\lg \hmax^{\text{final}}} O( (\lg k + \lg \ell) 2^i)
		&= O(\lg k + \lg \ell) \sum_{i=0}^{\lg \hmax^{\text{final}}} 2^i \\
		&= O( (\lg k + \lg \ell) \cdot \hmax^{\text{final}}) \\
		&= O( (\lg k + \lg \ell) \cdot \hmax^*).
	\end{align*}
	Hence, the cost paid by the above procedure which does not know $\hmax^*$ in
	advance is asymptotically the same as that of the procedure which knows
	$\hmax^*$ in advance.

\subsection{Proof of Lemma~\ref{onl:lem:marking-scheme}}
\label{onl:sec:proof:lem:marking-scheme}
	Before we prove the lemma, we first need to prove another claim and another
	lemma.
	\begin{claim}
	\label{onl:cla:shift}
	Suppose we have an ILP solution $x$ with $x_{(r,m)}>0$ for a configuration
	$(r,m)$ with $r\not\ge m$ but $r\ge_p m$. Then we can reduce the cost of $x$
	by at least $1+\lambda\operatorname{id}(m)$.
	\end{claim}
	\begin{proof}
		We change the extraordinary configuration $(r,m)$ into
		an ordinary configuration without creating additional extraordinary
		configurations (i.e., for all $m'\neq m$: $\sum_{r: r\not\ge m'}x_{(r,m')}$ will stay
		the same). We change entries of $r$ in a step by step process starting with the
		highest coordinate. Suppose that we already have $r_i\ge m_i$ for $i>j$ and
		that still $r\ge_p m$. Assume that $r_j= m_j-\xi $. Then we set $r_j^{\mathrm{new}}:= m_j$ and
		$r_{j-1}^{\mathrm{new}}:=r_{j-1}-\xi$ as the new entries for coordinate $j$ and $j-1$,
		respectively. Note that $\|r\|_1$ does not change. However, the ILP now has
		\begin{equation}
		\label{onl:constraints}
		\sum\nolimits_{(\bar{r},\bar{m})}x_{(\bar{r},\bar{m})}\bar{r}_{j-1}\ge V_{j-1}-\xi \text{~~~and~~~}
		\sum\nolimits_{(\bar{r},\bar{m})}x_{(\bar{r},\bar{m})}\bar{r}_{j}\ge V_{j}+\xi\enspace.
		\end{equation}
		Since
		$\sum_{(\bar{r},\bar{m})}x_{(\bar{r},\bar{m})}\bar{m}_j\le V_j$ there
		must exist a configuration $(r',m')$
		such that $x_{(r',m')}>0$ and $r_j'>m_j'$. We choose such a configuration,
		decrease $r_j'$ and increase $r_{j-1}'$ by the same amount. This does not
		change $\|r'\|_1$ and we can choose the increment so that $r_j'\ge m_j'$ still
		holds. Repeating this process can fix the constraints in Eq.~(\ref{onl:constraints})
		without generating new extraordinary configurations.

		Fixing all coordinates in $r$ results in an ordinary configuration for the
		source vector $m$ and this will reduce the cost of the ILP by $1+\lambda\operatorname{id}(m)$.
	\end{proof}

	\begin{lemma}
	\label{onl:lem:largeconfextraordinary}
		Suppose the optimal ILP solution uses an extraordinary configuration for
		some source vector~$m$, i.e., $x_{(r,m)}>0$ with $r\not\ge m$.
		Then it does not use any ordinary configurations of the form $(r',m')$,
		$m'\ge_p m$, i.e., $\sum_{r\ge m'}x_{(r,m')}=0$.
	\end{lemma} 
	\begin{proof}
		Assume for contradiction that the lemma does not hold. Let $m'\ge_p m$,
		where $m$ and $m'$ are source vectors with reservation $r$ and $r'$,
		respectively. Further assume that $r\not\ge m$ and $r'\ge m'$.

		We switch the reservation vector between the configurations, i.e., we increase
		$x_{(r',m)}$ and $x_{(r,m')}$ and decrease $x_{(r,m)}$ and $x_{(r',m')}$.

		If $r'\ge m$ then this step decreased the cost of the ILP because
		we decreased the number of extraordinary configurations of the form $(\cdot,
		m)$ and we (may) have increased the number of extraordinary configurations
		of the form $(\cdot, m')$. This decreases the overall cost and
		contradicts that $x$ is an optimal ILP solution.

		Now assume $r'\not\ge m$. We may have increased the cost of the ILP solution
		by $1+\lambda\operatorname{id}(m')$.
		However, $r'\ge m'\ge_p m$ implies $r'\ge_p m$. This means that we can reduce
		the cost of this solution by $1+\lambda\operatorname{id}(m)$ due to
		Claim~\ref{onl:cla:shift}. Altogether the cost decreases because
		$\operatorname{id}(m)>\operatorname{id}(m')$. This contradicts the
		fact that $x$ is an optimal ILP solution.
	\end{proof}

	Now we prove Lemma~\ref{onl:lem:marking-scheme}.
	First, suppose that $|S_m|\le\hmax$. Then all servers with source vector $m$
	are marked and the lemma clearly holds.
	Otherwise, let $X_{m}:= \{s\mid s\in S_m; m_s\neq m\}$. Since the paging
	problem for $S_m$ must leave at least $\hmax$ pages from $S_m$ outside the
	cache, there are at least $\hmax-|X_m|$ of these that have source vector
	$m$. The ILP has at most $\hmax$ extraordinary configurations. If at least one
	server with source vector $m$ is extraordinary (i.e., $h_m>0$) then all servers
	in $X_m$ are assigned an extraordinary configuration due to Lemma~\ref{onl:lem:largeconfextraordinary}.
	Hence, $h_m\le \hmax-|X_m|$ and the lemma follows.

\subsection{Proof of Claim~\ref{onl:cla:markingBoundedByPaging}}
	Initially, we have to choose for every paging problem $\hmax$ servers/pages
	that are not in the cache. The marking scheme marks these servers and
	experiences a cost of $1$ for each marking. This gives a total cost
	of $|M|\hmax$ for the initialization.

	As long as $|S_m|>\hmax$, a monochromatic merge-operation on a server with
	source vector $m$ introduces a request of weight $|\ps|$. If the marking scheme
	has to pay for the merge (because the corresponding server $s$ is
	extraordinary) then the page $s$ is outside of the cache in the paging problem
	for $S_m$ and the paging problem has to pay for the request. The total extra
	charge for monochromatic merge-operations that occur after $|S_m|\le\hmax$ can
	be at most $O(\log k)\cdot\hmax$ because the total extra charge for a specific
	server $s$ is at most $O(\log k)$ due to Claim~\ref{onl:cla:extraChargeServer}.

\subsection{Proof of Claim~\ref{onl:cla:pagingBoundedByhmax}}
	Note that an offline paging algorithm for the constructed request sequence on
	$S_m$ can simply determine the $\hmax$ elements that leave the set $S_m$ last
	and put all other servers into the cache. It then experiences a cost of at most
	$\hmax+O(\log k)/r\cdot\hmax=O(\hmax)$, where the first $\hmax$-term is due to
	the initialization. The second term comes from the fact that the total weight
	of all requests to a specific page is equal (up to constant factors) to the
	total extra charge for a specific server. The latter is at most $O(\log k)$ due
	to Claim~\ref{onl:cla:extraChargeServer}. Hence, by not moving any page after the
	initialization $\OPT$ pays $O(\log k)/r\cdot\hmax$ for serving the page
	requests.

	Since $r=\log k$ and since the online algorithm of Blum et al.\ is
	$O(r+\log\ell)$-competitive, we obtain
	$\cost(S_m)\le O(\log k+\log\ell)\cdot\hmax$.

\subsection{Proof of Theorem~\ref{onl:thm:randomized}}
	Analyzing the cost for the vertex charges is identical to the deterministic
	case. Combining Observation~\ref{onl:obs:markingscheme} with
	Claim~\ref{onl:cla:markingBoundedByPaging} and Claim~\ref{onl:cla:pagingBoundedByhmax}
	gives that the expected total extra charge is only $O(\log \ell+\log
	k)\cdot\hmax$. As $\cost(\OPT)=\Omega(\hmax)$ the theorem follows.

\subsection{Claim~\ref{onl:claim:picking-delta}}
\label{onl:sec:claim:picking-delta}
\begin{claim}
\label{onl:claim:picking-delta}
	Let $\varepsilon\in(0,\frac{1}{4})$ and $k\geq 10/\varepsilon^4$. Then there exists $\delta$
	such that (1)~$\frac{1}{2}\epsilon^2\le\delta\le\epsilon^2$,
	(2)~$\delta=i\frac{1}{k}$ with $i\in\mathbb{N}$ and
	(3)~$\lceil 1\rceil_\delta - 1 \leq \delta/2$, where
	$\lceil\cdot\rceil_\delta$ is the operation of rounding up to a multiple of
	$\delta$.
\end{claim}
\begin{proof}
	Set $\delta^* = \max\{ i \frac{1}{k} : i \frac{1}{k} \leq
	\varepsilon^2 \}$ and observe that $\varepsilon^2 - \frac{1}{k} < \delta^*
	\leq \varepsilon^2$.  Thus, $\delta^*$ satisfies~(1) and~(2).
	Next, if~(3) is satisfied, we have found a value $\delta = \delta^*$ with
	the desired properties.  Otherwise, we set $\delta=\delta^*$ and then we
	keep on decreasing $\delta$ by $\frac{1}{k}$ until~(3) holds. Clearly, when
	this procedure stops,~(2) and~(3) are satisfied.  It remains to show
	that~(1) holds.

	Observe that for every $\delta$, we have that
	$\lceil 1 \rceil_\delta = \lceil \frac{1}{\delta} \rceil \delta$, where
	$\lceil\cdot\rceil$ denotes rounding up to the next integer.
	Note that when the above procedure decreases $\delta$ by $\frac{1}{k}$,
	we still have $\lceil \frac{1}{\delta} \rceil = \lceil \frac{1}{\delta^*} \rceil$:
	After the first decrease, the value of $\lceil 1 \rceil_{\delta^*}$ drops by
	$\lceil \frac{1}{\delta^*} \rceil \cdot \frac{1}{k}
	\leq \frac{2}{\varepsilon^2} \cdot \frac{\varepsilon^4}{10}
	\leq \frac{\varepsilon^2}{5}$. Thus, for $\delta = \delta^*-\frac{1}{k}$ we
	still have $\lceil 1 \rceil_\delta - 1 \geq 0$. Using that
	$\delta\leq\delta^*$ implies $\lceil \frac{1}{\delta} \rceil \geq
	\lceil\frac{1}{\delta^*}\rceil$, gives the claim for
	$\delta=\delta^*-\frac{1}{k}$. Now the argument can be extended using
	induction.

	Now we show that~(1) holds when the procedure finishes. Observe that every time
	the procedure decreases $\delta$ by $\frac{1}{k}$, 
	the RHS of~(3) decreases by $\frac{1}{2k} \leq \frac{\varepsilon^4}{20}$ and
	the LHS of~(3) decreases by
	$\lceil \frac{1}{\delta} \rceil \cdot \frac{1}{k} = \lceil
	\frac{1}{\delta^*} \rceil \cdot \frac{1}{k}
	\geq \frac{1}{2\varepsilon^2} \cdot \frac{1}{k} \gg \frac{1}{k}$.
	Since we have that $\lceil 1 \rceil_{\delta^*} - 1 - \delta^*/2 \leq
	\varepsilon^2/2$ and we just saw that the procedure decreases the LHS much
	faster than $\delta$, the above procedure finishes with
	$\delta \geq \frac{1}{2}\varepsilon^2$.
\end{proof}

\subsection{Optimality of Monochromatic Merges and Commits on Ordinary Servers}
\label{onl:sec:ilp-solution-still-optimal}
Let $c$ be a cost vector for the ILP from Section~\ref{onl:sec:ilp}. Note that $c$
has an entry $c_{(r,m)}$ for each $\gamma$-feasible configuration $(r,m)$.
For the deterministic algorithm we simply have
\begin{equation}
	c_{(r,m)}:=
	\begin{cases}
		0, & r\geq m, \\
		1, & r\not\geq m,
	\end{cases}
\end{equation}
and for the randomized algorithm $c$ is defined as in
Equation~\ref{onl:eq:cost-vector}.

Recall the definition of the partial order $\leq_p$ from
Section~\ref{onl:sec:augmented-ilp}. We say that a cost vector $c$ is
\emph{$\leq_p$-respecting} if it satisfies (1) $c_{(r,m)} = 0$ for all ordinary
configurations $(r,m)$ and (2)~$m_1 \geq_p m_2 \Longrightarrow
c_{(r_1,m_1)} \leq c_{(r_2,m_2)}$ for all extraordinary configurations
$(r_1,m_1)$ and $(r_2,m_2)$.  It is easy to see that the cost vectors for the
deterministic and the randomized algorithm both satisfy this condition.

\begin{lemma}
\label{onl:lem:ilp-solution-still-optimal}
	Suppose the ILP is equipped with a $\leq_p$-respecting cost vector.
	Let $s$ be a server and let $\ps$ and $\pl$ be two large pieces monochromatic
	for $s$. Further, let $x$ be an optimal ILP solution before the pieces $\ps$
	and $\pl$ get merged and let $x'$ be an optimal ILP solution directly after
	the merge. Then the objective function value of the ILP for $x'$ is not
	smaller than the objective function value of the ILP for $x$.
\end{lemma}
\begin{proof}
	We prove the lemma by contradiction. Suppose the new solution $x'$ has a
	smaller objective function value than $x$. We show that this implies that
	the solution $x$ before merging pieces $\ps$ and $\pl$ was not optimal.

	Let $\is$, $\il$ and $\im$ denote the size classes of $\ps$, $\pl$ and $\pm$.
	Now consider a schedule $S'$ respecting the ILP solution $x'$. 
	Let $s'$ be the server on which $\pm$ is scheduled (note that in $S'$ it might be the case
	that $s'\neq s$) and let $(r'_s,m'_s)$, $(r'_{s'},m'_{s'})$ denote the
	configurations of $s$ and $s'$ in $S'$, respectively.
	Now we set $S^*$ to the same schedule as $S'$ except that piece $\pm$ is
	replaced by $\ps$ and $\pl$.
	Note that in $S^*$ all server configurations are the same as in $S'$ except
	that $s'$ has reservation $r^*_{s'}$ and $s$ has source vector $m^*_s$ with
	\begin{equation*}
		\begin{array}{lclll}
			r_{s'\is}^*    &:= &r_{s'\is}'    &+ &|\ps|_c   \\
			r_{s'\il}^*    &:= &r_{s'\il}'    &+ &|\pl|_c\\
			r_{s'\im}^*    &:= &r_{s'\im}'    &- &|\pm|_c   \\
		\end{array}
	\end{equation*}
	and
	\begin{equation*}
		\begin{array}{lclll}
			m^*_{s{\is}}    &:= &m'_{s{\is}}    &+ &|\ps|_c\\
			m^*_{s{\il}}    &:= &m'_{s{\il}}    &+ &|\pl|_c\\
			m^*_{s{\im}}    &:= &m'_{s{\im}}    &- &|\pm|_c
		\end{array}.
	\end{equation*}
	It follows that $S^*$ and $S'$ have exactly the same extraordinary servers.

	Now based on $S^*$ we derive an ILP solution $x^*$ by setting $x^*_{(r,m)}$
	to the number of servers in $S^*$ with configuration $(r,m)$ for all
	$(r,m)$. It is easy to see that $x^*$ is a feasible solution for the ILP
	before $\ps$ and $\pl$ were merged.  Now we distinguish two cases.  First,
	suppose $s$ is not extraordinary in $S^*$. Then $x^*$ has the same objective
	function value for the ILP as $x'$. This contradicts the optimality of $x$
	before $\ps$ and $\pl$ were merged.  Second, suppose that $s$ is
	extraordinary in $S^*$. In this case note that we have $m_s' \leq_p m^*_s$
	and hence $c_{(r'_s,m_s')} \geq c_{(r^*_s,m_s^*)}$ since the cost vector $c$
	is $\leq_p$-presering. Thus, the objective function value of $x^*$ is upper
	bounded by the objective function of $x'$. This again contradicts the
	optimality of $x$.
\end{proof}

\begin{lemma}
\label{onl:lem:ilp-solution-still-optimal-commit}
	Suppose the ILP is equipped with a $\leq_p$-respecting cost vector.
	Let $s$ be a server and let $\pm$ be a large monochromatic
	pieces for $s$. Let $x$ be an optimal ILP solution before volume is
	committed for $\pm$
	and let $x'$ be an optimal ILP solution directly after
	the commit. Then the objective function value of the ILP for $x'$ is not
	smaller than the objective function value of the ILP for $x$.
\end{lemma}
\begin{proof}
	We proceed similar to the proof of the
	Lemma~\ref{onl:lem:ilp-solution-still-optimal-commit}.
	We prove the lemma by contradiction. Suppose the new solution $x'$ has a
	smaller objective function value than $x$. We show that this implies that
	the solution $x$ before committing the volume for $\pm$ was not optimal.

	Let $\im$ and $\im'$ denote the size classe of $\pm$ before and after the
	commit, respectively.
	Now consider a schedule $S'$ respecting the ILP solution $x'$. 
	Let $s'$ be the server on which $\pm$ is scheduled (note that in $S'$ it might be the case
	that $s'\neq s$) and let $(r'_s,m'_s)$, $(r'_{s'},m'_{s'})$ denote the
	configurations of $s$ and $s'$ in $S'$, respectively.
	Now we set $S^*$ to the same schedule as $S'$ except that we undo the commit
	for piece $\pm$.
	Note that in $S^*$ all server configurations are the same as in $S'$ except
	that $s'$ has reservation $r^*_{s'}$ and $s$ has source configuration $m^*_s$ with
	\begin{equation*}
		\begin{array}{lclll}
			r_{s'0}^*		&:= &r_{s'0}'		&+ &\delta   			\\
			r_{s'\im}^*		&:= &r_{s'\im}'		&+ &|\pm|_c - \delta	\\
			r_{s'\im'}^*	&:= &r_{s'\im'}'	&- &|\pm|_c   			\\
		\end{array}
	\end{equation*}
	and
	\begin{equation*}
		\begin{array}{lclll}
			m^*_{s0}		&:= &m'_{s0}    	&+ &\delta\\
			m^*_{s{\im}}    &:= &m'_{s{\im}}    &+ &|\pm|_c -\delta\\
			m^*_{s{\im'}}   &:= &m'_{s{\im'}}   &- &|\pm|_c
		\end{array}.
	\end{equation*}
	It follows that $S^*$ and $S'$ have exactly the same extraordinary servers.

	Now based on $S^*$ we derive an ILP solution $x^*$ by setting $x^*_{(r,m)}$
	to the number of servers in $S^*$ with configuration $(r,m)$ for all
	$(r,m)$. It is easy to see that $x^*$ is a feasible solution for the ILP
	before the commit was performed.  Now we distinguish two cases.  First,
	suppose $s$ is not extraordinary in $S^*$. Then $x^*$ has the same objective
	function value for the ILP as $x'$. This contradicts the optimality of $x$
	before the commit.  Second, suppose that $s$ is
	extraordinary in $S^*$. In this case note that we have $m_s' \leq_p m^*_s$
	and hence $c_{(r'_s,m_s')} \geq c_{(r^*_s,m_s^*)}$ since the cost vector $c$
	is $\leq_p$-presering. Thus, the objective function value of $x^*$ is upper
	bounded by the objective function of $x'$. This again contradicts the
	optimality of $x$.
\end{proof}

\end{document}